\documentclass[11pt]{article}
\usepackage[]{graphicx}
\usepackage[]{color}
\usepackage{booktabs}
\usepackage{lscape}
\usepackage{alltt}
\usepackage{bbm}
\usepackage{color}
\usepackage[english]{babel}
\usepackage{graphics}
\usepackage{graphicx}
\usepackage{epsfig}
\usepackage{lscape}
\usepackage{amsmath,amsfonts,amssymb, amsthm}
\usepackage{eurosym}
\usepackage{natbib}
\usepackage{url}
\usepackage{hyperref}
\usepackage{setspace}
\usepackage{caption}
\usepackage{subcaption}
\usepackage{arydshln}
\usepackage{rotating}
\usepackage[title]{appendix}
\usepackage{todonotes}
\usepackage{pdflscape}
\usepackage{etoolbox}
\usetikzlibrary{patterns}
\AtBeginEnvironment{quote}{\small}
\usepackage{algorithm}
\usepackage{algpseudocode}
\usetikzlibrary{matrix, positioning, backgrounds}
\usetikzlibrary{arrows, decorations.markings,positioning}

\newtheorem{theorem}{Theorem}

\newtheorem{corollary}[theorem]{Corollary}
\newtheorem{proposition}[theorem]{Proposition}

\newtheorem{definition}[theorem]{Definition}
\newtheorem{remark}[theorem]{Remark}
\def\@begintheorem#1#2{\trivlist \item[\hskip \labelsep{\bf #1\ #2.}]\sl}
\renewenvironment{abstract}
 {\begin{center}\normalsize\textsc{Abstract}%
 \end{center}\begin{quote}\normalsize}
 {\end{quote}}

\usepackage{bm}

\renewcommand{\epsilon}{\varepsilon}

\newcommand\y{{\bm y}}

\newcommand\x{{\bm x}}

\usepackage{tikz}

\bibpunct{(}{)}{;}{a}{}{}
\begin{document}
\doublespacing

\title{Individual Treatment Effect: Prediction Intervals and Sharp Bounds}
\author{Zhehao Zhang and Thomas S. Richardson\\
University of Washington, Seattle \\
Department of Statistics}

\maketitle
\begin{abstract}
Individual treatment effect (ITE) is often regarded as the ideal target of inference in causal analyses and has been the focus of several recent studies. In this paper, we describe the intrinsic limits regarding what can be learned concerning ITEs given data from large randomized experiments. We consider when a valid prediction interval for the ITE is informative and when it can be bounded away from zero. The joint distribution over potential outcomes is only partially identified from a randomized trial. Consequently, to be valid, an ITE prediction interval must be valid for all joint distribution consistent with the observed data and hence will in general be wider than that resulting from knowledge of this joint distribution. We characterize prediction intervals in the binary treatment and outcome setting, and extend these insights to models with continuous and ordinal outcomes. We derive sharp bounds on the probability mass function (pmf) of the individual treatment effect (ITE). Finally, we 
contrast prediction intervals for the ITE and confidence intervals for the average treatment effect (ATE). This also leads to the consideration of Fisher versus Neyman null hypotheses. While confidence intervals for the ATE shrink with increasing sample size due to its status as a population parameter, prediction intervals for the ITE generally do not vanish, leading to scenarios where one may reject the Neyman null yet still find evidence consistent with the Fisher null, highlighting the challenges of individualized decision-making under partial identification.
\end{abstract}
\section{Introduction}
The traditional causal inference literature has been focused on population level treatment effect parameters such as the average treatment effect (ATE) and conditional average treatment effect (CATE). Although the individual treatment effect (ITE) is often regarded as the ideal parameter of interest for personalized decision making, it is not, in general identifiable even if we know the outcome for an individual and have data from a large randomized experiment. Recent works by \cite{lei2021conformal}, \cite{jin2023sensitivity}, \cite{chernozhukov2023toward},\cite{wang2025conformal} discuss conformal inference methods to estimate ITE and prediction intervals for the ITE. Another recent debate from \cite{mueller2022personalized} and \cite{dawid2023personalised} also explores the possibility of using ITE and bounds on ITE to help personalized decision making. 

Bounds and relationships on the probability of the counterfactual treatment effect has been studied in terms of probability of causation in \cite{robins1989probability} and \cite{tian2000probabilities}. Some more recent work from \cite{mueller2021causes,sani2023bounding,kawakami2025mediation} extends these ideas to learn individual responses and bounds from causal diagrams and mediation analysis. Inference and examples of probability of causation have been discussed in \cite{dawid2016statistical}. \cite{fan2010sharp} study sharp bounds on the cdf of the individual treatment effect using copulas based on the previous result from \cite{frank1987best} and \cite{williamson1990probabilistic}. \cite{mullahy2018individual} applies the bounds in \cite{fan2010sharp} in health economics applications. Bounds and optimal policy for binary treatment and outcome have been studied in  \cite{kallus2022treatment}, \cite{kallus2022s}, and \cite{dawid2023personalised}. Individual treatment effects on ordinal outcomes have been studied in \cite{lu2018treatment}. Numerical approaches for solving causal inference problems in discrete settings have been explored in \cite{duarte2024automated}. Exact inference on individual treatment effects based on permutation has been studied in \cite{blaker2000confidence,rigdon2015randomization,chiba2015exact}. The construction of confidence intervals for causal parameters has been discussed in \cite{robins1988confidence}, \cite{imbens2018causal} and \cite{brennan2024causal}.

In this work, we try to understand the prediction intervals and bounds for the ITE given that we have observations from well-conducted large randomized control trials (RCT). In Section \ref{sec:ITE_prediction_binary}, we start from the binary treatment and outcome model and provide a complete characterization of ITE prediction intervals under this simple model setting. We characterize when a degenerate interval consisting of a single value is a valid prediction interval, when is the valid prediction interval for ITE non-negative/non-positive, and when the only valid prediction interval is a trivial interval. Additionally, we give conditions on the observed data for there to exist a consistent joint distribution under which a given non-trivial prediction interval would be valid.
In Section~\ref{sec:beyond_binary}, we extend our insights to continuous and ordinal outcomes. Our approach leverages cdf bounds for the difference of two random variables from \cite{fan2010sharp} and \cite{zhang2024bounds}.
In Section \ref{sec:binary_Frechet}, we provide sharp bounds on the cumulative distribution function (cdf) and probability mass function (pmf) of the ITE under binary treatment and outcome model. We explore the relationships between the cdf/pmf bounds on ITE and the Fr\'{e}chet-Hoeffding bounds on two random variables. This motivates the general sharp bounds on the pmf of the ITE for ordinal outcome in section \ref{sec:sharp_PMF_bounds}. 
Lastly, we compare ITE prediction intervals and ATE confidence intervals and provide a synthetic data example to discuss the implications of them in section \ref{sec:discuss_ate_ite}.


\section{The limits to inference for individual treatment effects in binary treatment and outcome model}\label{sec:ITE_prediction_binary}

Throughout the paper, we consider a binary treatment setting with treatment $D=0,1$. Let $Y_1$ be the potential outcome when receiving the treatment and $Y_0$ be the potential outcome when not receiving the treatment. We define the individual treatment effect as:
\begin{align*}
\text{ITE} = Y_1-Y_0.
\end{align*}
Let $Y$ be the observed variable. By consistency, $Y=Y_0$ when $D=0$ and $Y=Y_1$ when $D=1$. In this section, we will examine what is possible to learn in the limit of a large sample size.
\subsection{Definition of a prediction interval for the individual treatment effect (ITE)}

A $(1-\alpha)$ {\em prediction interval} for an individual treatment effect is an interval such that
\[
P\left((Y_1-Y_0) \in [L,R]\right) \geq 1-\alpha.
\]
Throughout the paper, we assume $\alpha$ is sufficiently bounded away from $0.5$.

\begin{proposition}
     Suppose an interval $\mathbb{I}$ is a valid $(1-\alpha)\%$ prediction interval. If $P(\mathrm{ITE}\in A)>\alpha$ for some set $A$, then $\mathbb{I}\cap A \neq \emptyset$. If $P(\mathrm{ITE}\in A)\leq \alpha$ for some set $A$, then there exists a $(1-\alpha)$ valid prediction set that does not intersect with $A$. In particular, if $\mathbb{R}\setminus A$ is an interval, then it will be a valid $(1-\alpha)$ prediction interval.
\end{proposition}
\begin{proof}
    If $P(\mathrm{ITE}\in A)> \alpha$ and $\mathbb{I}\cap A = \emptyset$, then $P(\mathrm{ITE}\in\mathbb{I})\leq 1-P(\mathrm{ITE}\in A)<1-\alpha$. The interval $\mathbb{I}$ does not have $1-\alpha$ coverage. If $P(\mathrm{ITE}\in A)\leq \alpha$, then $P(\mathrm{ITE}\in (\mathbb{R}\setminus A))= 1- P(\mathrm{ITE}\in A)\geq 1-\alpha$. $\mathbb{R}\setminus A$ is a valid $(1-\alpha)$ prediction set.
\end{proof}
\begin{corollary}
    In discrete settings, whenever $P(ITE=i)>\alpha$, we must have $i\in \mathbb{I}$.  
\end{corollary}

\subsection{Binary Treatment and Outcome Model}

To further narrow the discussion, we first consider the case in which the treatment and response are both binary. Following  \cite{copas1973randomization} who characterizes individual patients in the binary treatment and outcome model, we have the following four types and individual treatment effects:
\begin{center}
\begin{tabular}{cccc}
$Y_{0}$	&$Y_{1}$	&ITE & Type\\
\hline
0	&0	&0	& Never Recover (NR)\\
0	&1	&1	&Helped (HE)\\
1	&0	&-1	&Hurt (HU)\\
1	&1	&0	&Always Recover (AR) / Immune \\
\end{tabular}
\end{center}
Notice that in this setting the individual treatment effects take three possible values: $-1,0,1$.

 In the simple setting that we consider, since there are only $3$ possible values taken by the ITE, there are $6$ possible prediction intervals:
\begin{equation*}
\{-1\},\;\; 
\{0\},\;\; 
\{1\},\;\; 
[-1,0],\;\; 
[0,1],\;\; 
[-1,1]. 
\end{equation*}

Note that prediction intervals $[-1,0],[0,1],[-1, 1]$ here are sets correspond to $\{-1,0\}, \{0,1\}, \{-1,0,1\}$. Singleton sets can be viewed as degenerate intervals where the starting point equals to the end point. In this simple setting, there is only one prediction set that does not correspond to an interval, namely $\{-1,1\}$. We discussed it in Remark \ref{re:-1_1_not_possible}.

\subsubsection{When will the valid prediction interval of minimal length for a given joint distribution not be unique?}
\begin{corollary}\label{cor:non-unique}
If $P(HE\cup HU)>\alpha$ and $\max\{P(HE),P(HU)\}\leq \alpha$, then $\{0\}$ is not a valid $(1-\alpha)\%$ prediction interval but both $[-1,0]$ and $[0,1]$ are valid and minimal length $(1-\alpha)\%$ prediction intervals for the ITE.
\end{corollary}

Under the conditions stated in Corollary~\ref{cor:non-unique}, the set $\{0\}$ fails to provide sufficient coverage; yet each of $[-1,0]$ and $[0,1]$ captures a sufficiently large portion of the probability mass (at least $1-\alpha$).  Both of these intervals therefore qualify as valid $(1-\alpha)\%$ prediction intervals for $\mathrm{ITE}$, and each achieves the same minimal length (one unit).  Consequently, there is no single ``shortest'' interval that strictly dominates the other. In general, even given the joint distribution of $Y_0, Y_1$, there could be multiple minimal-length valid prediction intervals.

\subsubsection{Partial Identification of the Joint Distribution over Types under Randomization}

Under randomization, the relationship between the counterfactual distribution
$P(Y_0, Y_1)$ and the observed distributions $\left\{P(Y\mid D=0), P(Y\mid D=1)\right\}$ 
 is given by this table:
 
\begin{center}
\begin{tabular}{r|cc}
 & $P(Y\!=\!0\mid D\!=\!0)$ & $P(Y\!=\!1\mid D\!=\!0)$\\[4pt]
 \hline
 &\\[-8pt]
 $P(Y\!=\!0\mid D\!=\!1)$ & $P( Y_0\!=\!0, Y_1\!=\!0)$ & $P( Y_0\!=\!1, Y_1\!=\!0)$\\[4pt]

 $P(Y\!=\!1\mid D\!=\!1)$ & $P( Y_0\!=\!0, Y_1\!=\!1)$ & $P( Y_0\!=\!1, Y_1\!=\!1)$\\[4pt]
\end{tabular}
\end{center}
\medskip

Here $P(Y\!=\!i\mid D\!=\!j) = P( Y_j\!=\!i)$ due to randomization.

\bigskip

Equivalently we may write this in terms of types:

\begin{center}
\begin{tabular}{r|cc}
 & $P(Y\!=\!0\mid D\!=\!0)$ & $P(Y\!=\!1\mid D\!=\!0)$\\[4pt]
 \hline
 &\\[-8pt]
 $P(Y\!=\!0\mid D\!=\!1)$ & $P($NR$)$ & $P($HU$)$ \\[4pt]
  $P(Y\!=\!1\mid D\!=\!1)$ & $P($HE$)$ &$P($AR$)$ \\[4pt]
\end{tabular}
\end{center}

\bigskip
\begin{proposition}[Fr\'{e}chet inequality bounds]\label{def:frechet_ineq}

 For two real valued random variables $Y_0, Y_1$ and any $(y_0,y_1)\in \mathbb{R}^2$, suppose that we know $P(Y_0=y_0)=a$ and $P(Y_1=y_1)=b$, then
\begin{align*}
\max\{0,a+b-1\}\leq P(Y_0=y_0,Y_1=y_1)\leq \min\{a,b\}
\end{align*}
This is also known as the Boole-Fr\'{e}chet inequality.
\end{proposition}

Since under randomization the joint distribution must add up to satisfy the observed distributions in the two treatment arms, we can parametrize the distribution of types using $P(AR)=t$. We then have the following solution set:

\begin{equation*}\label{eq:jointwithmargins}
\left\{ \begin{array}{ccl}
P(\hbox{AR})&=&{\textcolor{red}{ t}},\\
P(\hbox{HU})&=&P(Y\!=\!1 \mid D\!=\!0)-{\textcolor{red}{ t}},\\
P(\hbox{HE})&=&P(Y\!=\!1 \mid D\!=\!1)-{\textcolor{red}{ t}},\\
P(\hbox{NR}) &=& 1-P(Y\!=\!1 \mid D\!=\!0)-P(Y\!=\!1 \mid D\!=\!1) +{\textcolor{red}{ t}},
\end{array}\!\!\!\
\right\},
\end{equation*}
where, by Proposition \ref{def:frechet_ineq}, we have
\begin{align*}
{\textcolor{red}{ t}} \geq \max\left\{0,(P(Y\!=\!1 \mid D\!=\!0)+P(Y\!=\!1 \mid D\!=\!1)) - 1\right\},\\
{\textcolor{red}{ t}} \leq \min\left\{P(Y\!=\!1 \mid D\!=\!0),P(Y\!=\!1 \mid D\!=\!1)\right\}.
\end{align*}

\subsection{When is the only valid prediction interval trivial?\label{trivial_pred}}

There are circumstances in which the only valid prediction interval for the individual treatment effect is the trivial interval: $[-1,1]$!

The interval will be trivial when the set of people of type Hurt can be larger than $\alpha$, {\em and} the set of people of type Helped  can also be larger than $\alpha$.

Notice that the proportion Helped and Hurt, {\em both } achieve their maximum value when the proportion Always Recover 
${\textcolor{red}{ t}}$ achieves its minimum value. Hence the only valid prediction interval will be trivial whenever we have both:
\begin{align*}
P(Y\!=\!1 \mid D\!=\!0) - \max\left\{0,(P(Y\!=\!1 \mid D\!=\!0)+P(Y\!=\!1 \mid D\!=\!1)) - 1\right\}>\alpha;\\
P(Y\!=\!1 \mid D\!=\!1) - \max\left\{0,(P(Y\!=\!1 \mid D\!=\!0)+P(Y\!=\!1 \mid D\!=\!1)) - 1\right\} > \alpha.
\end{align*}
this may be equivalently expressed as:
\begin{align*}
 \min\left\{P(Y\!=\!1 \mid D\!=\!0),1-P(Y\!=\!1 \mid D\!=\!1)\right\} &> \alpha;\\
 \min\left\{P(Y\!=\!1 \mid D\!=\!1),1-P(Y\!=\!1 \mid D\!=\!0)\right\} &> \alpha.
\end{align*}

Consequently, provided we have:
\begin{align*}
\alpha < \min\{P(Y\!=\!1 \mid D\!=\!0), P(Y\!=\!1 \mid D\!=\!1)\}
\end{align*}
and
\begin{align*}
\max\{P(Y\!=\!1 \mid D\!=\!0), P(Y\!=\!1 \mid D\!=\!1)\}< (1-\alpha)
\end{align*}
then the only valid prediction interval will be trivial.
In concrete terms, if $\alpha=0.05$ and the proportions recovering under treatment ($D=1$) and control ($D=0$) both
lie between $5\%$ and $95\%$ then, given that we know the conditional distributions
$P(Y|D)$,  the only valid 95\% prediction interval for the individual treatment effect will be $[-1,1]$.

These results show that for randomized experiments in which the proportion of recovery in both arms lies between $\alpha$ and $(1-\alpha)$, the observed data is entirely uninformative regarding the ITE.

\subsection{When is the valid prediction interval for the individual treatment effect a singleton?\label{singleton}}
Notwithstanding the results in the previous section, perhaps surprisingly, there are situations in which a valid $(1-\alpha)\%$ prediction interval is a singleton. We now characterize when this occurs. Detailed derivations are provided in the Appendix~\ref{app:when_valid}. There are three cases to consider:
\begin{itemize}
    \item $\{0\}$ is valid if the sum of AR and NR types is at least $1-\alpha$, i.e., both arms have nearly $\alpha\%$ or nearly $1-\alpha\%$ response such that the sum of the proportions of individuals across the two arms with the less common outcome is less than $\alpha$. Concretely, either
\begin{align*}
P(Y\!=\!1 \mid D\!=\!0)+ P(Y\!=\!1 \mid D\!=\!1)\leq \alpha
\end{align*}
or 
\begin{align*}
P(Y\!=\!0 \mid D\!=\!0)+ P(Y\!=\!0 \mid D\!=\!1)\leq \alpha.
\end{align*}
Note however, that in such a case the average treatment effect, though less than $\alpha$, may be non-zero if $P(Y=0\mid D=0)\neq P(Y=0\mid D=1)$. 
Thus, given sufficiently large sample sizes, confidence intervals for the Average Treatment Effect will not include zero. We note that the singleton ITE prediction interval $\{0\}$ is equivalent to establishing the Fisherian sharp null hypothesis \citep{fisher1936design} holds for at least $(1-\alpha)$ of the population. We will further discuss this in section~\ref{sec:discuss_ate_ite}. 
    \item $\{1\}$ is valid if the HE type alone can exceed $1-\alpha$. This requires
    \begin{align*}
    P(Y=1\mid D=1) - P(Y=1\mid D=0) \geq (1-\alpha).
    \end{align*}
    meaning the average treatment effect is at least $1-\alpha$.
    \item $\{-1\}$ is valid if the HU type alone can exceed $1-\alpha$. This requires
    \begin{align*}
    P(Y=1\mid D=0) - P(Y=1\mid D=1) \geq (1-\alpha).
    \end{align*}
    i.e., the average treatment effect is less than $-(1-\alpha)$.
\end{itemize}
\subsection{When is the valid prediction interval for the individual treatment effect non-negative/non-positive?\label{non_neg_pos}}
Following the previous result, one can rule out negative (or positive) treatment effects if the proportion of HU (or HE) can never exceed $\alpha$. This occurs when either $P(Y=1 \mid D=0)$ or $1 - P(Y=1 \mid D=1)$ is below $\alpha$ (and similarly for ruling out positive effects). See Appendix~\ref{app:when_valid}.

\begin{remark}[]\label{re:-1_1_not_possible}
It is not possible to conclude $\{-1,1\}$ as a best prediction set given the observed marginals $P(Y=1\mid D=1), P(Y=1\mid D=0)$. For  $\{-1,1\}$ to be a valid prediction set, the proportion of people of type Helped plus the proportion of people of type Hurt should always be greater than or equal to $1-\alpha$. From the lower bounds on proportion of people of type Helped and the proportion of people of type Hurt, we need
\begin{equation*}
\begin{aligned}
&P(Y=1\mid D=1)+P(Y=1\mid D=0)
\\&-2\min\{P(Y=1\mid D=1),P(Y=1\mid D=0)\}\geq 1-\alpha,
\end{aligned}
\end{equation*}
equivalently, either
\begin{align*}
P(Y=1\mid D=1)-P(Y=1\mid D=0)\geq 1-\alpha
\end{align*}
or 
\begin{align*}
P(Y=1\mid D=0)-P(Y=1\mid D=1)\geq 1-\alpha.
\end{align*}
However, based on the result in Section \ref{singleton}, in either setting, we would just conclude the singleton $\{-1\}$ or $\{1\}$ respectively as the prediction set instead of $\{-1,1\}$. Therefore, it is not possible to conclude $\{-1,1\}$ as a best prediction set from the observed marginals $P(Y=1\mid D=1), P(Y=1\mid D=0)$. Note that $\{-1,1\}$ is a theoretically possible prediction set which can be optimal for certain disributions over potential outcomes. It is just we will not be able to conclude it based on observed distribution from randomization.
\end{remark}

\subsection{Visual summary of results on valid ITE prediction intervals}
Based on the result in \ref{trivial_pred},  \ref{singleton}, and \ref{non_neg_pos}, under randomization and in the limit of a large sample size where we observe the true marginals $P(Y=1\mid D=0), P(Y=1\mid D=1)$ and their counterparts, we can characterize the corresponding ITE prediction intervals. These intervals are valid no matter the joint distribution over $Y_0$ and $Y_1$ (provided it is compatible with $P(Y\mid D)$. These intervals are also ``the best we can do" without additional assumptions or information, e.g. from a cross-over study. 

\begin{figure}[h!]
  \centering
  \includegraphics[width=0.65\linewidth]{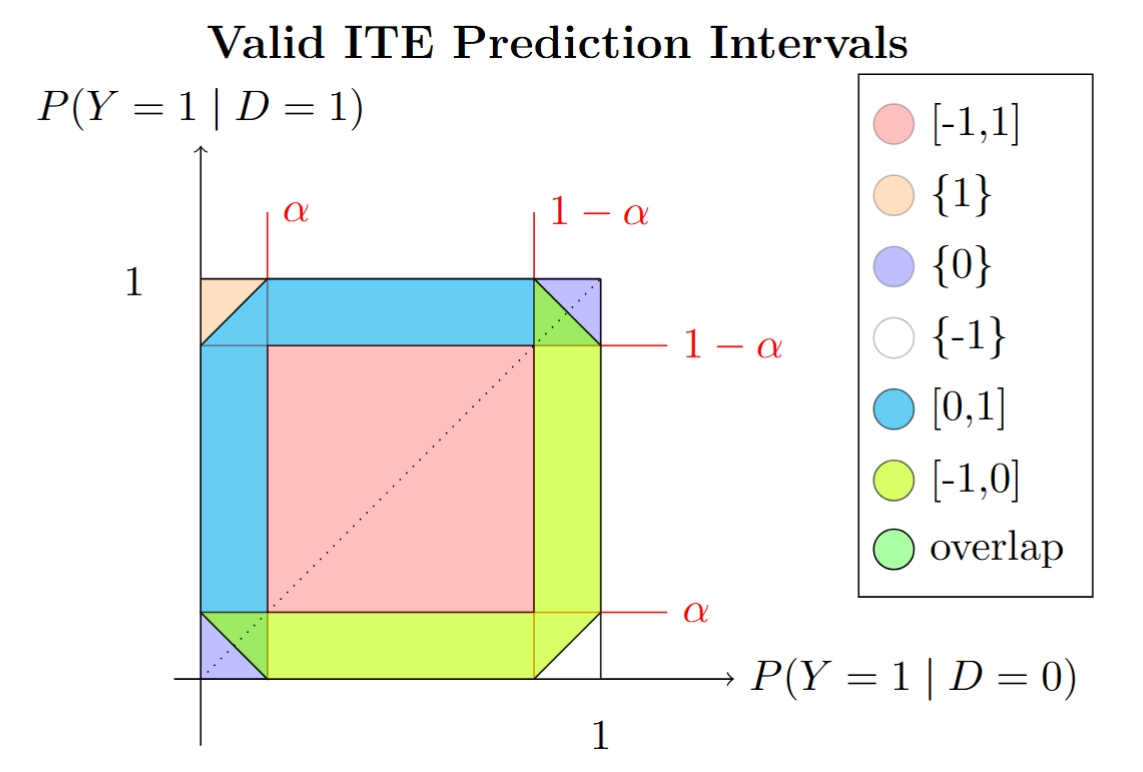}
\caption{Shortest ITE intervals given different marginal distributions.\label{F1}}
\end{figure}

Notice that we have two overlapping triangles each with area $\frac{1}{2}\alpha^2$ in Figure \ref{F1} where both $[0,1]$ and $[-1,0]$ can serve as valid $\alpha-$level prediction interval for individual treatment effects. Assume that for the same length of intervals, intervals with higher coverage are "better".  Then we can further decompose these triangular areas to obtain the "best" prediction intervals; see Figure \ref{F2}.

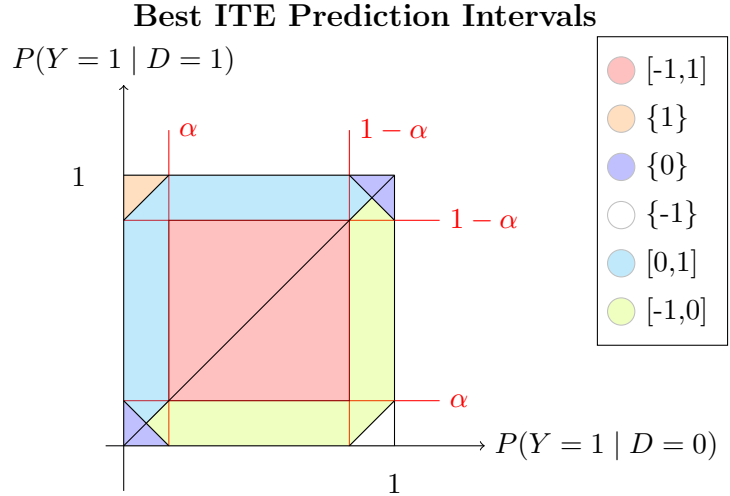
\begin{figure}[h!]
\centering
\begin{tikzpicture}[domain=0:3.5, scale=1.2]
  \draw[semithick]
 (-0.5,3.2) node[below]{$1$};
  \draw[semithick]
 (3,-0.2) node[below]{$1$};
  \draw[->] (-0.2,0) -- (4,0) node[right] {$P(Y=1\mid D=0)$};
  \draw[->] (0,-0.5) -- (0,4) node[above] {$P(Y=1\mid D=1)$};
  \draw[-] (3,0) -- (3,3) node[below] {};
  \draw[-] (0,3) -- (3,3) node[below] {};
  \draw[color=red]    plot (\x,0.5)             node[right] {$\alpha$};
  \draw[color=red]    plot (0.5,\x)             node[right] {$\alpha$};
  \draw[color=red]    plot (2.5,\x)             node[right] {$1-\alpha$};
  \draw[color=red]    plot (\x,2.5)             node[right] {$1-\alpha$};

  \draw[domain=0:0.5]    plot (\x,-\x+0.5)  node[below]{};
  \draw[domain=0:0.5]    plot (\x,\x+2.5)  node[below]{};
  \draw[domain=2.5:3]    plot (\x,-\x+5.5)  node[below]{};
  \draw[domain=2.5:3]    plot (\x,\x-2.5)  node[below]{};
  \draw[fill=red,nearly transparent]  (0.5,0.5) -- (0.5,2.5) -- (2.5,2.5) -- (2.5,0.5) --(0.5,0.5) -- cycle;
  \draw[fill=blue,nearly transparent]  (0,0) -- (0.5,0)-- (0,0.5) -- cycle;
  \draw[fill=blue,nearly transparent]  (2.5,3) -- (3,2.5)-- (3,3) -- cycle;
  \draw[fill=white,nearly transparent]  (2.5,0) -- (3,0)-- (3,0.5) -- cycle;
  \draw[fill=orange,nearly transparent]  (0,3) -- (0,2.5)-- (0.5,3) -- cycle;
\draw[fill=cyan,nearly transparent] (0,0.5)--(0.25,0.25)--  (0.5,0.5) -- (0.5,2.5) -- (2.5,2.5) -- (2.75,2.75)-- (2.5,3)--(0.5,3)--(0,2.5) -- cycle;
\draw[fill=lime,nearly transparent] (0.5,0.5) --(0.25,0.25) -- (0.5,0)-- (2.5,0)--(3,0.5) --(3,2.5)--(2.75,2.75)--(2.5,2.5)--(2.5,0.5) -- cycle;
  \draw[domain=0:3]    plot (\x,\x)  node[below]{};
  \matrix [draw,below left] at (current bounding box.north east) {
    \node [shape=circle, draw=black,fill=red,nearly transparent,label=right:{[-1,1]}] {}; \\
  \node [shape=circle, draw=black,fill=orange,nearly transparent,label=right:{\{1\}}] {}; \\
    \node [shape=circle, draw=black,fill=blue,nearly transparent,label=right:{\{0\}}] {}; \\
    \node [shape=circle, draw=black,fill=white,nearly transparent,label=right:{\{-1\}}] {}; \\
    \node [shape=circle, draw=black,fill=cyan,nearly transparent,label=right:{[0,1]}] {}; \\
    \node [shape=circle, draw=black,fill=lime,nearly transparent,label=right:{[-1,0]}] {}; \\
};
\node[above,font=\large\bfseries] at (current bounding box.north) {Best ITE Prediction Intervals};
\end{tikzpicture}
\caption{The prediction interval with the highest coverage among those that are valid and have minimal length.}\label{F2}
\end{figure}
\subsection{Necessary conditions for a given prediction interval to be valid and of minimal length}\label{subsec:binary_necessary}
Suppose we are given a prediction interval and an observed distribution from an RCT, we can consider when does there exist a joint distribution over potential outcomes that is compatible with the observed distributions and for the prediction interval to be valid and of minimal length. 
In Appendix~\ref{app:necessary_conditions}, we derive the precise constraints on the marginal distributions 
\(\,P(Y=1 \mid D=0)\) and \(P(Y=1 \mid D=1)\) 
that guarantee each of the intervals can serve as a valid (or “best”) prediction interval for the individual treatment effect for some joint distribution over potential outcome.
The necessary conditions take the form of simple inequalities relating the two response probabilities; they characterize scenarios in which each the true treatment effect with probability at least \(1-\alpha\) lies within the given set for som population compatible with $P(Y\mid D)$. Full details, derivations, and proofs appear in Appendix~\ref{app:necessary_conditions}.

\subsection{Visualization of Necessary Conditions given an ITE prediction interval to be the best}
Figure \ref{figure:marginal1} and \ref{figure:marginal2} provide visualizations of the necessary conditions on $P(Y=1\mid D=0)$ and $P(Y=1\mid D=1)$ for a given prediction interval to be valid/best. Note that if for a given point in the unit square, we consider the intervals (from Figure~\ref{figure:marginal1} and \ref{figure:marginal2}) containing $(P(Y=1\mid D=1), P(Y=1\mid D=0))$ and then select the longest, we recover Figure~\ref{F2}.

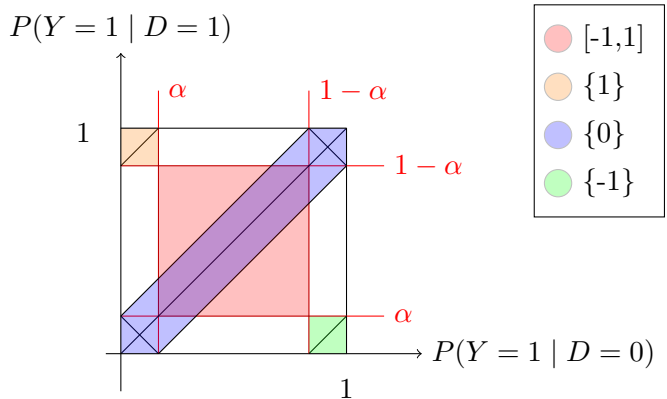
\begin{figure}[h!]
\centering
\begin{tikzpicture}[domain=0:3.5]
  \draw[semithick]
 (-0.5,3.2) node[below]{$1$};
  \draw[semithick]
 (3,-0.2) node[below]{$1$};
  \draw[->] (-0.2,0) -- (4,0) node[right] {$P(Y=1\mid D=0)$};
  \draw[->] (0,-0.5) -- (0,4) node[above] {$P(Y=1\mid D=1)$};
  \draw[-] (3,0) -- (3,3) node[below] {};
  \draw[-] (0,3) -- (3,3) node[below] {};
  \draw[color=red]    plot (\x,0.5)             node[right] {$\alpha$};
  \draw[color=red]    plot (0.5,\x)             node[right] {$\alpha$};
  \draw[color=red]    plot (2.5,\x)             node[right] {$1-\alpha$};
  \draw[color=red]    plot (\x,2.5)             node[right] {$1-\alpha$};

  \draw[domain=0:0.5]    plot (\x,-\x+0.5)  node[below]{};
  \draw[domain=0:3]    plot (\x,\x)  node[below]{};
  \draw[domain=0:2.5]    plot (\x,\x+0.5)  node[below]{};
  \draw[domain=0.5:3]    plot (\x,\x-0.5)  node[below]{};
  \draw[domain=0:0.5]    plot (\x,\x+2.5)  node[below]{};
  \draw[domain=2.5:3]    plot (\x,-\x+5.5)  node[below]{};
  \draw[domain=2.5:3]    plot (\x,\x-2.5)  node[below]{};
  \draw[fill=red,nearly transparent]  (0.5,0.5) -- (0.5,2.5) -- (2.5,2.5) -- (2.5,0.5) --(0.5,0.5) -- cycle;
   \draw[fill=blue,nearly transparent]  (0,0) -- (0.5,0)--(3,2.5)--(3,3)--(2.5,3)-- (0,0.5) -- cycle;
  \draw[fill=green,nearly transparent]  (2.5,0) -- (3,0)-- (3,0.5)--(2.5,0.5) -- cycle;
  \draw[fill=orange,nearly transparent]  (0,3) -- (0,2.5)--(0.5,2.5)--  (0.5,3) -- cycle;
  \matrix [draw,below left] at (current bounding box.north east) {
    \node [shape=circle, draw=black,fill=red,nearly transparent,label=right:{[-1,1]}] {}; \\
  \node [shape=circle, draw=black,fill=orange,nearly transparent,label=right:{\{1\}}] {}; \\
    \node [shape=circle, draw=black,fill=blue,nearly transparent,label=right:{\{0\}}] {}; \\
    \node [shape=circle, draw=black,fill=green,nearly transparent,label=right:{\{-1\}}] {}; \\
};
\end{tikzpicture}
\caption{Necessary condition on the marginal distributions for a given prediction interval, respectively, $[-1,1], \{1\}, \{0\}, \{-1\}$, to be valid and best for some joint distribution $P(Y_0,Y_1)$ compatible with $P(Y\mid D)$; see also \ref{figure:marginal2}.\label{figure:marginal1}}
\end{figure}

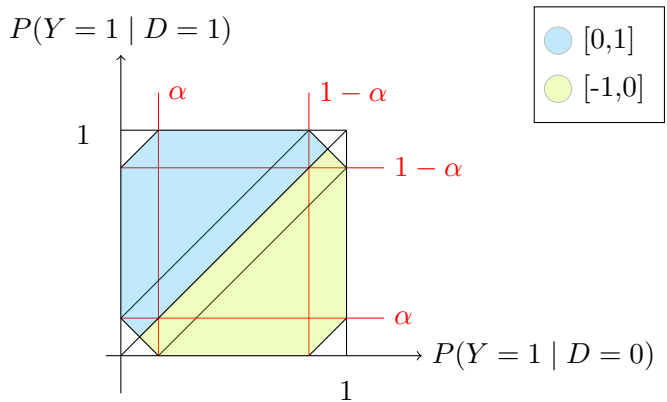
\begin{figure}[h!]
\centering
\begin{tikzpicture}[domain=0:3.5]
  \draw[semithick]
 (-0.5,3.2) node[below]{$1$};
  \draw[semithick]
 (3,-0.2) node[below]{$1$};
  \draw[->] (-0.2,0) -- (4,0) node[right] {$P(Y=1\mid D=0)$};
  \draw[->] (0,-0.5) -- (0,4) node[above] {$P(Y=1\mid D=1)$};
  \draw[-] (3,0) -- (3,3) node[below] {};
  \draw[-] (0,3) -- (3,3) node[below] {};
  \draw[color=red]    plot (\x,0.5)             node[right] {$\alpha$};
  \draw[color=red]    plot (0.5,\x)             node[right] {$\alpha$};
  \draw[color=red]    plot (2.5,\x)             node[right] {$1-\alpha$};
  \draw[color=red]    plot (\x,2.5)             node[right] {$1-\alpha$};

  \draw[domain=0:0.5]    plot (\x,-\x+0.5)  node[below]{};
  \draw[domain=0:3]    plot (\x,\x)  node[below]{};
  \draw[domain=0:2.5]    plot (\x,\x+0.5)  node[below]{};
  \draw[domain=0.5:3]    plot (\x,\x-0.5)  node[below]{};
  \draw[domain=0:0.5]    plot (\x,\x+2.5)  node[below]{};
  \draw[domain=2.5:3]    plot (\x,-\x+5.5)  node[below]{};
  \draw[domain=2.5:3]    plot (\x,\x-2.5)  node[below]{};
\draw[fill=cyan,nearly transparent] (0,0.5)--(0.25,0.25) -- (2.75,2.75) --  (2.5,3)--(0.5,3)--(0,2.5) -- cycle;
\draw[fill=lime,nearly transparent]  (0.25,0.25) -- (0.5,0)-- (2.5,0)--(3,0.5) --(3,2.5)--(2.75,2.75) -- cycle;
  \matrix [draw,below left] at (current bounding box.north east) {
    \node [shape=circle, draw=black,fill=cyan,nearly transparent,label=right:{[0,1]}] {}; \\
    \node [shape=circle, draw=black,fill=lime,nearly transparent,label=right:{[-1,0]}] {}; \\
};
\end{tikzpicture}
\caption{Necessary condition on the marginal distributions for a given prediction interval $[0,1], [-1,0]$ to be valid and best for some joint distribution $P(Y_0,Y_1)$ compatible with $P(Y\mid D)$. \label{figure:marginal2}}
\end{figure}

\section{Beyond binary outcomes}\label{sec:beyond_binary}
We now turn to the case where the outcome is continuous, rather than binary. Our goal is to address questions analogous to those explored in the binary setting. In particular, we begin by investigating how to construct valid prediction intervals based solely on the marginal distributions, without imposing assumptions on the joint distribution of the potential outcomes. We then examine the conditions under which these intervals can be bounded away from zero. Further, we solve the problem: given a prediction interval (or set), what conditions on the observed data ensure the existence of a consistent joint distribution under which the interval is valid?
\subsection{A conservative interval for continuous outcomes}\label{subsec:conti_pi}
The marginal distributions of $Y_1$ and $Y_0$ are identified from randomized experiments.
Let $[L_0, R_0]$ be an interval such that $P\left(Y_0\in [L_0,R_0]\right) \geq 1-\alpha/2.$ Let $[L_1, R_1]$ be an interval such that $P\left(Y_1\in [L_1,R_1]\right) \geq 1-\alpha/2.$ Then
\begin{align}
   P\left((Y_1-Y_0) \in [L_1-R_0, R_1-L_0]\right) \geq 1-\alpha. \label{eq:conservative_int}
\end{align}
\begin{proof}
    Note that if $Y_1-Y_0$ is not in $[L_1-R_0, R_1-L_0]$, then either $Y_1$ is not in $[L_1, R_1]$ or $Y_0$ is not in $[L_0,R_0]$. Put into set notation,
\[
\{Y_1 - Y_0 \notin [L_1 - R_0,\; R_1 - L_0] \}
\subseteq
\{Y_1 \notin [L_1,R_1]\} 
\cup
\{Y_0 \notin [L_0,R_0]\}.
\]
Applying the union bound to that set containment gives
\[
P\Bigl(Y_1 - Y_0 \notin [L_1 - R_0, R_1 - L_0]\Bigr)
\leq
P\bigl(Y_1 \notin [L_1,R_1]\bigr) + P\bigl(Y_0\notin [L_0,R_0]\bigr).
\]
By hypothesis, $P\left(Y_0\notin [L_0,R_0]\right) < \alpha/2$ and $P\left(Y_1\notin [L_1,R_1]\right) < \alpha/2$, therefore,
\[P\bigl(Y_1-Y_0 \notin [L_1,R_1]\bigr) < \alpha.\]
Thus,
$$P\left((Y_1-Y_0) \in [L_1-R_0, R_1-L_0]\right) \geq 1-\alpha.$$
\end{proof}
This result is similar to the ``naive" conformal inference prediction interval for the ITE given in \cite{lei2021conformal} Section 4.1. Though we further assume an infinite sample size and do not use covariates. 
\subsection{Points that must be included in every valid $(1-\alpha)$ prediction interval}\label{subsec:points_include}
One might argue that the bounds in (\ref{eq:conservative_int}) are too conservative. We will take a different perspective to see what are the points that must be included in the interval. Note that to be valid, an ITE prediction interval must be valid for all joint distributions consistent with the observed data, and hence will in general be wider than that resulting from knowledge of this joint distribution. Let $[L_0', R_0']$ be an interval such that 
\[L_0' 
:=
\min \Bigl\{\ell \in \mathbb{R} : P\bigl(Y_0 < \ell\bigr) > \alpha\Bigr\}\]
and
\[R_0' 
:=
\max \Bigl\{\ell \in \mathbb{R} : P\bigl(Y_0 > \ell\bigr) > \alpha\Bigr\}.\] In other words, $L'_0$ and $R'_0$ are the $\alpha$-quantile and $(1-\alpha)$-quantile of $Y_i(0)$. Similarly we can define
\[L_1' 
:=
\min \Bigl\{\ell \in \mathbb{R} : P\bigl(Y_1 < \ell\bigr) > \alpha\Bigr\}\]
and
\[R_1' 
:=
\max \Bigl\{\ell \in \mathbb{R} : P\bigl(Y_1 > \ell\bigr) > \alpha\Bigr\}.\]
Then a valid $(1-\alpha)$ prediction interval for the ITE must include these points: 
\begin{itemize}
  \item $R_1'-L_0'$,
  \item $L_1'-R_0'$.
\end{itemize}
This follows because otherwise there exists a joint distribution of $Y_1, Y_0$ such that there is more than $\alpha$ mass outside of the prediction interval. This is because the only constraint imposed on the joint distribution by the marginals is given by the Fréchet inequalities,
\[P(Y_1> R_1', Y_0<L_0')\leq \min\{P(Y_1> R_1'), P(Y_0<L'_0)\}\]
By construction the minimum is greater than $\alpha$, meaning there exists a joint distribution that $P(Y_1> R_1', Y_0<L_0')>\alpha$. Hence under this joint distribution, $P(Y_1-Y_0>R_1'-L_0')>\alpha$, therefore, any interval that does not include $R_1'-L_0'$ will not be a valid $\alpha$ level prediction interval. 

Even small tails in marginal distributions of $Y_1, Y_0$ can force the prediction interval to expand considerably at both extremes. We illustrate this in Figure~\ref{fig:continuous_y_case}.

\begin{figure}[h!]
    \centering
\resizebox{0.75\textwidth}{!}{
\tikzset{every picture/.style={line width=0.75pt}} 

\begin{tikzpicture}[x=0.75pt,y=0.75pt,yscale=-1,xscale=1]

\draw    (226,140.55) -- (465.7,140.5) ;
\draw    (226,140.55) -- (226.2,290) ;
\draw    (231.2,127.5) .. controls (356.7,72) and (331.2,-22.9) .. (462.2,129.1) ;
\draw    (207.5,284.6) .. controls (38.5,200.55) and (167.5,176.6) .. (207.5,146.6) ;
\draw [color={rgb, 255:red, 251; green, 1; blue, 1 }  ,draw opacity=1 ]   (447.2,76.1) -- (449.3,288.15) ;
\draw [color={rgb, 255:red, 251; green, 1; blue, 1 }  ,draw opacity=1 ]   (154.7,155.45) -- (496.3,155.15) ;
\draw    (203.8,125.15) -- (203.98,149.6) ;
\draw [shift={(204,151.6)}, rotate = 269.57] [color={rgb, 255:red, 0; green, 0; blue, 0 }  ][line width=0.75]    (10.93,-3.29) .. controls (6.95,-1.4) and (3.31,-0.3) .. (0,0) .. controls (3.31,0.3) and (6.95,1.4) .. (10.93,3.29)   ;
\draw    (468.3,95.15) -- (455.97,124.75) ;
\draw [shift={(455.2,126.6)}, rotate = 292.61] [color={rgb, 255:red, 0; green, 0; blue, 0 }  ][line width=0.75]    (10.93,-3.29) .. controls (6.95,-1.4) and (3.31,-0.3) .. (0,0) .. controls (3.31,0.3) and (6.95,1.4) .. (10.93,3.29)   ;
\draw [color={rgb, 255:red, 74; green, 144; blue, 226 }  ,draw opacity=1 ]   (452.8,141.15) -- (448.3,156.15) ;
\draw [color={rgb, 255:red, 74; green, 144; blue, 226 }  ,draw opacity=1 ]   (457.3,140.65) -- (452.8,155.65) ;
\draw [color={rgb, 255:red, 74; green, 144; blue, 226 }  ,draw opacity=1 ]   (461.3,141.15) -- (456.8,156.15) ;
\draw [color={rgb, 255:red, 74; green, 144; blue, 226 }  ,draw opacity=1 ]   (465.7,140.5) -- (461.2,155.5) ;
\draw [color={rgb, 255:red, 251; green, 1; blue, 1 }  ,draw opacity=1 ]   (156.2,269.45) -- (497.8,269.15) ;
\draw [color={rgb, 255:red, 251; green, 1; blue, 1 }  ,draw opacity=1 ]   (241.7,84.95) -- (243.8,297) ;
\draw [color={rgb, 255:red, 74; green, 144; blue, 226 }  ,draw opacity=1 ]   (243.2,270) -- (226.3,276.65) ;
\draw [color={rgb, 255:red, 74; green, 144; blue, 226 }  ,draw opacity=1 ]   (243.2,275.35) -- (226.3,282) ;
\draw [color={rgb, 255:red, 74; green, 144; blue, 226 }  ,draw opacity=1 ]   (243.7,279.85) -- (226.8,286.5) ;
\draw    (502.4,183.85) -- (450.08,157.06) ;
\draw [shift={(448.3,156.15)}, rotate = 27.11] [color={rgb, 255:red, 0; green, 0; blue, 0 }  ][line width=0.75]    (10.93,-3.29) .. controls (6.95,-1.4) and (3.31,-0.3) .. (0,0) .. controls (3.31,0.3) and (6.95,1.4) .. (10.93,3.29)   ;
\draw    (272.4,240.85) -- (244.62,268.59) ;
\draw [shift={(243.2,270)}, rotate = 315.05] [color={rgb, 255:red, 0; green, 0; blue, 0 }  ][line width=0.75]    (10.93,-3.29) .. controls (6.95,-1.4) and (3.31,-0.3) .. (0,0) .. controls (3.31,0.3) and (6.95,1.4) .. (10.93,3.29)   ;

\draw (174,130.65) node [anchor=north west][inner sep=0.75pt]  [font=\scriptsize]  {$-Y_{0}$};
\draw (472.5,115.65) node [anchor=north west][inner sep=0.75pt]  [font=\scriptsize]  {$Y_{1}$};
\draw (465.8,82.4) node [anchor=north west][inner sep=0.75pt]  [font=\scriptsize]  {$\textcolor[rgb]{0.97,0.02,0.02}{\alpha \%}$};
\draw (194.3,112.4) node [anchor=north west][inner sep=0.75pt]  [font=\scriptsize]  {$\textcolor[rgb]{0.97,0.02,0.02}{\alpha \%}$};
\draw (467.9,180.5) node [anchor=north west][inner sep=0.75pt]   [align=left] {{\scriptsize Need to include this value!}};
\draw (251.9,221) node [anchor=north west][inner sep=0.75pt]   [align=left] {{\scriptsize Need to include this value!}};

\end{tikzpicture}
}
\caption{Illustration of the continuous $Y$ case. In order to maintain a $1-\alpha$ coverage probability for the ITE, the prediction interval must include the key quantile differences on both the left and right tails of the outcome distributions.}
\label{fig:continuous_y_case}
\end{figure}
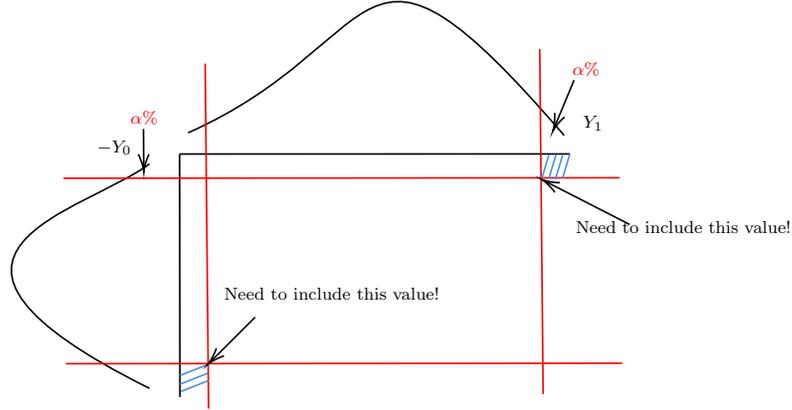

\subsection{Can we obtain a prediction interval bounded away from zero?}\label{subsec:away_0_continuous}
In the previous section, we gave conditions under which a point is included in every valid ITE prediction interval. Here we now give necessary and sufficient conditions for a valid interval to exclude the half line $(0,\infty)$ (or $(-\infty,0)$). Though the arguments extend to any constant, not just zero, we now address the question: can a valid prediction interval for the individual treatment effect (ITE) be bounded away from zero? Suppose there exists a joint distribution of the potential outcomes such that $P(Y_1 - Y_0 \leq 0) > \alpha$. Then, without further assumptions on the joint distribution, the left endpoint of any valid $(1-\alpha)$ prediction interval must be less than or equal to zero. This is closely related to Kolmogorov's problem and the results in \cite{fan2010sharp, zhang2024bounds}, which we review in Appendix~\ref{sec:sharp_CDF_bounds}. In particular, consider the following upper bound on $P(Y_1 - Y_0 \leq 0)$:
\begin{align}
    F^U(0) &= 1 + \inf_y \min\left\{F_1(y) - P(Y_0 < y),\, 0\right\} \nonumber\\
           &= 1 + \inf_y \min\left\{F_1(y) - F_0(y) + P(Y_0 = y),\, 0\right\} \label{eq:upp_bound_0}.
\end{align}
If this upper bound exceeds $\alpha$, then zero must lie within the prediction interval, i.e., the left endpoint cannot be strictly greater than zero.

Similarly, if $P(Y_1 - Y_0 \geq 0) > \alpha$, then the right endpoint of any valid prediction interval must be at least zero. This probability can be written as
\[
P(Y_1 - Y_0 \geq 0) = 1 - P(Y_1 - Y_0 < 0),
\]
so we may use a lower bound on $P(Y_1 - Y_0 < 0)$ to assess whether the right endpoint must include zero. The lower bound is given by:
 \begin{align}
F^L(0^-)&=\sup_y\max\{F_1(y)-F_0(y),0\}\label{eq:lower_bound_0-}.
\end{align}
If this lower bound is less than $1 - \alpha$, then the right endpoint must be at least zero.

Therefore, if both the left endpoint must be less than or equal to zero and the right endpoint must be greater than or equal to zero, the prediction interval necessarily includes zero and cannot be bounded away from it. 

In fact, if the distribution of $Y_1-Y_0$ is absolutely continuous with respect to Lebesgue measure so $P(Y_1-Y_0=0)=0$, then we can combine the conditions in (\ref{eq:upp_bound_0}) and (\ref{eq:lower_bound_0-}) to state that for a valid $(1-\alpha)$ prediction interval for the ITE to exclude zero, either the lower bound on $P(Y_1 - Y_0 \leq 0)$ to exceed $1 - \alpha$ or the upper bound to be less than $\alpha$—conditions that are only met in highly extreme cases.

In Appendix~\ref{section:high_dimension}, we briefly discuss the cases of bounding ITE prediction intervals in the ordinal outcome setting, using a new result that we will prove in Section \ref{sec:sharp_PMF_bounds}. 

\section{Fr\'{e}chet-Hoeffding bound on the pmf and cdf of ITE under binary treatment and outcome model}\label{sec:binary_Frechet}
In the previous section, we used bounds on the cumulative distribution function (cdf) of the individual treatment effect (ITE) under known marginal distributions for potential outcomes. We now turn to the corresponding problem for the probability mass function (pmf). As before, we begin with the binary treatment and outcome setting.
We first consider the binary treatment and binary outcome model where we know the marginals for $Y_0$ is $p, 1-p$ and the marginals for $Y_1$ is $q, 1-q$.
We can parameterize the joint density using $P(Y_0=1, Y_1=0)=t$ and Table $\ref{table_binary}$.
\begin{center}
\begin{table}[h]
\scalebox{0.85}{
\begin{tabular}{r|ccc}
 & $P(Y\!=\!0\mid X\!=\!0)=p$ & $P(Y\!=\!1\mid X\!=\!0)=1-p$ \\[4pt]
 \hline
 &\\[-5pt]
 $P(Y\!=\!0\mid X\!=\!1)=q$ & $t\in[\max\{0,p+q-1\},\min\{p, q\}]$& $q-t$\\[4pt]
  $P(Y\!=\!1\mid X\!=\!1)=1-q$ & $p-t$ &$1-q-p+t$ \\[4pt]
    \end{tabular}
\bigskip
}
\caption{Binary Treatment and Outcome Model}
\label{table_binary}
\end{table}
\end{center}

\begin{definition}[Fr\'{e}chet-Hoeffding bounds]
For two real valued random variables $Y_1,Y_0$ and any $y_1,y_0\in \mathbb{R}$, suppose that we know $P(Y_0\leq y_0)=a$ and $P(Y_1\leq y_1)=b$, then
\begin{align*}
\max\{0,a+b-1\}\leq P(Y_0\leq y_0, Y_1\leq y_1)\leq \min\{a,b\}.
\end{align*}
\end{definition}
Note: Here we make the distinction that Fr\'{e}chet-Hoeffding bounds are bounding the joint cdf of $Y_1, Y_0$ while the Fr\'{e}chet inequality bounds are bounding the joint pmf of $Y_1, Y_0$.

\begin{definition}[Comonotonicity]
The bivariate random vector $Y=(Y_0,Y_1)\in \mathbb{R}^2$ is called comonotonic if
\begin{align*}
P(Y_0\leq y_0, Y_1\leq y_1)=\min\{P(Y_0\leq y_0),P(Y_1\leq y_1)\}
\end{align*}
for all $y_0,y_1\in \mathbb{R}$. In this case, the bivariate distribution functions $F(y_1,y_0)=P(Y_0\leq y_0, Y_1\leq y_1)$ achieves the Fr\'{e}chet-Hoeffding upper bound and we call the two random variables $Y_0, Y_1$ perfectly positively dependent.
\end{definition}

\begin{definition}[Countermonotonicity]
The bivariate random vector $Y=(Y_0,Y_1)\in \mathbb{R}^2$ is called countermonotonic if
\begin{align*}
P(Y_0\leq y_0, Y_1\leq y_1)=\max\{0,P(Y_0\leq y_0)+P(Y_1\leq y_1)-1\}
\end{align*}
for all $y_0,y_1\in \mathbb{R}$. In this case, the bivariate distribution functions $F(y_1,y_0)=P(Y_0\leq y_0, Y_1\leq y_1)$ achieves the Fr\'{e}chet-Hoeffding lower bound and we call the two random variables $Y_0, Y_1$ perfectly negatively dependent.
\end{definition}
In the binary treatment and outcome model, a perfect positive dependence of $Y_1, Y_0$ is achieved by $t=\min\{p,q\}$ and a perfect negative dependence of $Y_1,Y_0$ is achieved by $t=\max\{0,p+q-1\}$. 

\begin{proposition}
When the treatment and outcome are both binary, the sharp pmf bounds on $Y_1-Y_0$ are achieved at the Fr\'{e}chet-Hoeffding bounds. \label{prop:PMF_binary}
\end{proposition}

From Table \ref{table_binary}, we have $P(\text{ITE}=-1)=q-t$, $P(\text{ITE}=0)=1-p-q+2t$, $P(\text{ITE}=1)=p-t$. The upper and lower bounds for each ITE value are reached at the Fr\'{e}chet-Hoeffding bound on the joint distribution of $Y_1, Y_0$ (i.e. when $t$ takes minimum or maximum value). 

We further observe that for $P(\text{ITE}=1)$, $P(\text{ITE}=-1)$, we can obtain pmf bounds by applying Fr\'{e}chet inequality bounds on each cell directly. And for $P(\text{ITE}=0)$, we can first obtain Fr\'{e}chet inequality bounds on the two cells $t\in[\max\{0,p+q-1\},\min\{p, q\}]$ and $1-q-p+t\in[\max\{0,1-p-q\},\min\{1-p, 1-q\}]$ and then add up the lower and upper bounds to obtain $1-q-p+2t\in [\max\{0,p+q-1\}+\max\{0,1-p-q\},\min\{p, q\}+\min\{1-p, 1-q\}]$. This simplifies to the bounds given by $1-p-q+2t, t\in [\max\{0,p+q-1\},\min\{p, q\}]$ as before. We will generalize this observation in section \ref{sec:sharp_PMF_bounds}.

\begin{proposition}
When the treatment and outcome are both binary, the sharp bounds on cdf of $Y_1-Y_0$ are achieved at the Fr\'{e}chet-Hoeffding bounds. \label{prop:CDF_binary}
\end{proposition}

Fr\'{e}chet-Hoeffding bound on the joint distribution of $Y_1, Y_0$ implies that $t\in[\max\{0,p+q-1\},\min\{p, q\}]$. Since we can parameterize the joint probability with one parameter $t$, we have $P(\text{ITE}=-1)=q-t$, $P(\text{ITE}=0)=1-p-q+2t$. In the binary case, the cdf for the ITE is characterized by the values $F(-1)=P(\text{ITE}=-1)=q-t$, $F(0)=P(\text{ITE}\leq 0)=P(\text{ITE}=-1)+P(\text{ITE}=0)=1-p+t$, $F(1)=1$. 

As noted in Section $2.1$ of \cite{fan2010sharp}, sharp bounds on the cdf of the individual treatment effect are not achieved at the Fr\'{e}chet-Hoeffding lower and upper bounds (perfectly positive/ perfectly negative dependence) for the distribution of $Y_1, Y_0$.  As a special case, the sharp bounds on cdf of ITE are achieved at the Fr\'{e}chet-Hoeffding bounds in binary treatment and outcome model.

\begin{proposition}\label{prop:CDF_to_PMF_binary}
In general, we can obtain valid pmf bounds from the cdf bounds. For example, if the ITE takes integer values, then $P(\mathrm{ITE}=i)= P(\mathrm{ITE}\leq i)-P(\mathrm{ITE}< i)=F(i)-F(i-1)$. If we know $F(i)\in [a,b]$ and $F(i-1)\in[c,d]$, we can obtain that 
\begin{align}
P(\mathrm{ITE}=i)\in [a-d, b-c].\label{eq:CDF_to_PMF}
\end{align}
When the treatment and outcome are both binary, the pmf bounds obtained from the sharp bounds on the cdf of $Y_1-Y_0$ using the above method are sharp.
\end{proposition}
Using the cdf bounds for $Y_1-Y_0$ in Proposition \ref{prop:CDF_binary}, we have $$P(\mathrm{ITE}=1)=F(1)-F(0),$$ with $F(1)=1, F(0)=1-p+t\in[1-p+\max\{0,p+q-1\},1-p+\min\{p,q\}]$. Threrefore, 
$$P(\mathrm{ITE}=1)\in [p-\min\{p,q\}, p-\max\{0,p+q-1\}].$$
Similarly, $P(\mathrm{ITE}=0)=F(0)-F(-1)$ for $F(-1)\in [q-\min\{p,q\}, q-\max\{0,p+q-1\}]$. Thus,
$$P(\mathrm{ITE}=0)\in[1-p-q+2\max\{0,p+q-1\},1-p-q+2\min\{p,q\}].$$ And 
$$P(\mathrm{ITE}=-1)=F(-1)=q-t, t\in[\max\{0,p+q-1\},\min\{p, q\}].$$ These bounds agree with the sharp bounds derived in Proposition \ref{prop:PMF_binary}. However, in general, the pmf bounds derived from cdf bounds using $(\ref{eq:CDF_to_PMF})$ are not sharp. Therefore, we will explore sharp pmf bounds in Section \ref{sec:sharp_PMF_bounds}.

\section{Sharp bounds on the pmf of Individual Treatment Effects}\label{sec:sharp_PMF_bounds}
A natural question to ask here is: given the marginals of $Y_1, Y_0$, can we derive sharp bounds on the probability mass function (pmf) of ITE, for example, what are the bounds for $P(\text{ITE}=0)?$

This is given by adding up the Fr\'{e}chet inequality bounds for each $P(Y_1=i, Y_0=i-\delta)$. 
\begin{theorem}\label{pmf}
For discrete random variables $Y_1, Y_0$ with given marginals and a given $\delta$, let $[L_i, U_i]=[\max\{P(Y_1=i)+P( Y_0=i-\delta)-1,0\},\min\{P(Y_1=i), P( Y_0=i-\delta)\}]$ be the Fr\'{e}chet inequality bounds for the joint probability $P(Y_1=i, Y_0=i-\delta)$. We claim that
\begin{align}
P(Y_1-Y_0=\delta)\in[\sum_i L_i, \sum_i U_i].\label{PMF_bound}
\end{align}
Furthermore, the bounds in (\ref{PMF_bound}) are sharp as there exists joint distributions of $Y_1,Y_0$ that satisfies the given marginals and achieves the upper or lower bound on $P(Y_1-Y_0=\delta)$.
\end{theorem}
\subsection{Proof of Theorem \ref{pmf}}
First, we show that the bound in (\ref{PMF_bound}) is a valid bound. Notice 
\begin{align*}
P(Y_1-Y_0=\delta)=\sum_iP(Y_1=i, Y_0=i-\delta).
\end{align*} 
And for all $i$, Fr\'{e}chet inequality bounds states that 
\begin{align*}
P(Y_1=i, Y_0=i-\delta)&\in[L_i, U_i],
\end{align*}
where
\begin{align*}
[L_i, U_i]&=[\max\{P(Y_1=i)+P( Y_0=i-\delta)-1,0\},\min\{P(Y_1=i), P( Y_0=i-\delta)\}].
\end{align*}
So the sum must be within the bound in $(\ref{PMF_bound})$.\\ 

Now, we will show that the bounds in $(\ref{PMF_bound})$ are tight in a sense that there exist joint distributions compatible with the marginal conditions that achieve the lower/upper bounds. First, we will state a very useful result in \cite{koperberg2024couplings}.

 Let $A$ and $B$ be sets and $R \subseteq A \times B$ a relation. Then for each $\mathrm{U} \subseteq A$ the set of neighbours of $\mathrm{U}$ in $R$ is denoted by
$$
\mathcal{N}_{\mathrm{R}}(\mathrm{U})=\{\mathrm{b} \in \mathrm{B}:(\mathrm{U} \times\{\mathrm{b}\}) \cap \mathrm{R} \neq \varnothing\}.
$$
\begin{theorem}[Strassen's theorem for finite sets] Let $\mathrm{A}$ and $\mathrm{B}$ be finite sets, $\mathrm{P}$ and $\mathrm{P}^{\prime}$ probability measures on $\mathrm{A}$ and $\mathrm{B}$ respectively and $\mathrm{R} \subseteq A \times \mathrm{B}$ a relation. Then there exists a coupling $\hat{\mathrm{P}}$ of $\mathrm{P}$ and $\mathrm{P}^{\prime}$ that satisfies $\hat{\mathrm{P}}(\mathrm{R})=1$ if and only if
$$
\mathrm{P}(\mathrm{U}) \leqslant \mathrm{P}^{\prime}\left(\mathcal{N}_{\mathrm{R}}(\mathrm{U})\right) \text {, for all } \mathrm{U} \subseteq A \text {. }
$$\label{strassen}
\end{theorem}
Now we can begin the proof with a proposition.
\begin{proposition}
If
\begin{align*}
&\max\{P(Y_1=i)+P( Y_0=j)-1,0\}>0\\
\text{and} \quad &\max\{P(Y_1=k)+P( Y_0=l)-1,0\}>0
\end{align*}
then we have either $i=k$ or $j=l$. In other words, at most one Fr\'{e}chet inequality lower bound can be non-zero if $i\neq k$ and $j\neq l$.
\label{prop:at_most_one}
\end{proposition} 
\begin{proof}
Suppose, for a contradiction that $i\neq j$, $k\neq l$ and
 \begin{align}
 P(Y_1=i)+P( Y_0=j)-1&>0 \label{s_p_1};\\ 
 P(Y_1=k)+P( Y_0=l)-1&>0\label{s_p_2}.
 \end{align}
 Summing up $(\ref{s_p_1})$ and $(\ref{s_p_2})$, we get
 \begin{align*}
  P(Y_1=i)+P( Y_0=j)+ P(Y_1=k)+P( Y_0=l)-2>0.
 \end{align*}
 However, we also have
 \begin{align*}
P(Y_1=i)+&P( Y_0=j)+P(Y_1=k)+P( Y_0=l)\\&\leq \sum_iP(Y_1=i)+\sum_jP(Y_0=j)=2.
\end{align*}
Therefore, we got a contradiction.
\end{proof}
We will now show that there exist a joint distribution of $Y_1, Y_0$ such that $P(Y_1=i, Y_0=i-\delta)=\max\{P(Y_1=i)+P( Y_0=i-\delta)-1,0\}=L_i$ for all $i$. By Proposition \ref{prop:at_most_one}, $P(Y_1=i, Y_0=i-\delta)\neq 0$ for at most one $i$.  There are thus two cases to consider here: when $L_i=0$ for all $i$ and when $L_i=0$ for all $i$ except one.\\

Case $1$: when $P(Y_1=i, Y_0=i-\delta)=0$ for all $i$. Case $1$ implies that
\begin{align}
&\max\{P(Y_1=i)+P( Y_0=i-\delta)-1,0\}=0 \quad \forall i\\
 \Rightarrow \quad& P(Y_1=i) \leq 1-P( Y_0=i-\delta) \quad \forall i .\label{lower_bound_zero_constraint}
\end{align}
We apply Theorem \ref{strassen}. Let $A,B$ be the support of the specified margins for $Y_1, Y_0$ respectively. Naturally, we have $P(Y_1), P(Y_0)$ as two probability measures on $A,B$. Let $R=\{(i,j):i\in A, j\in B, i-j\neq \delta\}$. Strassen's theorem states that there exists a coupling $\hat{P}$ of $P(Y_1), P(Y_0)$ with $\hat{P}(R)=1$ if and only if 
\begin{align*}
P(Y_1\in U) \leqslant P\left(Y_0\in\mathcal{N}_{R}(U)\right) \text {, for all } U \subseteq A ,
\end{align*}
where $$
\mathcal{N}_{R}(U)=\{b \in B:(U \times\{b\}) \cap R \neq \varnothing\}.
$$
By construction of $R$, for any $i\in A$,
\begin{align*}
\mathcal{N}_{R}(i) = B\setminus \{i-\delta\}.
\end{align*}
Therefore for any $U\subseteq A$ with more than one element, we have $\mathcal{N}_{R}(U)=B$ and:
\begin{align*}
\mathrm{P}\left(Y_0\in\mathcal{N}_{\mathrm{R}}(\mathrm{U})\right)=\mathrm{P}\left(Y_0\in B\right)=1.
\end{align*}
 Thus the only non trivial constraints are for singleton sets $U$ given by: 
\begin{align}
P(Y_1=i)\leq \sum_{j\neq i-\delta}P(Y_0=j)=1-P(Y_0=i-\delta). \label{lower_bound_zero_constraint_strassen}
\end{align}
The constraints in $(\ref{lower_bound_zero_constraint_strassen})$ are already satisfied by the assumption in  $(\ref{lower_bound_zero_constraint})$. Therefore, by Strassen's theorem, there must exist a coupling $P^\prime$ such that $P^\prime(R)=1$. Thus, the lower bounds $L_i=0$ for all $i$ are achievable in case $1$.\\

Case $2$: Consider the case that $L_i=0$ for all $i$ except for $i=j$. We will explicitly construct a joint distribution for $Y_1, Y_0$ that satisfies: 
\begin{align}
&P(Y_1=j, Y_0=j-\delta)=P(Y_1=j)+P( Y_0=j-\delta)-1 >0 \label{eq:lower_bound_not_0};\\ 
&P(Y_1=i, Y_0=i-\delta)=0 \label{eq:lower_bound_0_except_1}
\quad \text{ for all} \quad i\neq j.
\end{align} 
Consider a distribution satisfies by (\ref{eq:lower_bound_not_0}) and (\ref{eq:lower_bound_0_except_1}). Further, let $P(Y_1=i, Y_0=k)=0$ for any $i\neq j$ or $k\neq j-\delta$. Let $P(Y_1=j, Y_0=k)=P(Y_0=k), k\neq j-\delta$ and $P(Y_1=i, Y_0=j-\delta)=P(Y_1=i), i\neq j$.  For all $i\neq j$, we have
 \begin{align*}
\sum_kP(Y_1=i, Y_0=k)&= P(Y_1=i, Y_0=j-\delta)= P(Y_1=i).
 \end{align*}
For all $k\neq j-\delta$, we have
 \begin{align*}
\sum_iP(Y_1=i, Y_0=k)&= P(Y_1=j, Y_0=k)=P(Y_0=k).
 \end{align*}
Also, we have
 \begin{align*}
\sum_kP(Y_1=j, Y_0=k)&= P(Y_1=j, Y_0=j-\delta)+\sum_{k\neq j-\delta}P(Y_1=j, Y_0=k)\\
&=\max\{P(Y_1=j)+P( Y_0=j-\delta)-1,0\}+\sum_{k\neq j-\delta}P(Y_0=k)\\
&= P(Y_1=j)+P( Y_0=j-\delta)-1+(1-P(Y_0=j-\delta))\\
&=P(Y_1=j),
 \end{align*}
where the second equality follows definition of $j$. Similarly, 
  \begin{align*}
\sum_iP(Y_1=i, Y_0=j-\delta)&= P(Y_1=j, Y_0=j-\delta)+\sum_{i\neq j}P(Y_1=i, Y_0=j-\delta)\\
&=\max\{P(Y_1=j)+P( Y_0=j-\delta)-1,0\}+\sum_{i\neq j}P(Y_1=j)\\
&= P(Y_1=j)+P( Y_0=j-\delta)-1+(1-P(Y_1=j))\\
&=P(Y_0=j-\delta).
 \end{align*}
All the $P(Y_1=i, Y_0=j)$ are non-negative. This is a valid joint distribution that satisfies the given marginals. Thus, the lower bounds are achievable. An example of the construction of the joint probability matrix in case 2 is given in Figure \ref{fig:construction_matrix_case_2}. \\
\begin{figure}[h!]
\centering
\begin{tikzpicture}
\tikzset{
  rows/.style 2 args={
    sub@rows/.style={row ##1 column #2/.style={nodes={very thick, rectangle,draw=red}}},
    sub@rows/.list={#1}
  },
  box/.style 2 args={
    sub@box/.style={rows={#1}{##1}},
    sub@box/.list={#2}
  }
}
\tikzset{%
square matrix/.style={
    matrix of nodes,
    column sep=-\pgflinewidth, 
    row sep=-\pgflinewidth,
    nodes in empty cells,
    nodes={draw,
      minimum size=#1,
      anchor=center,
      align=center,
      inner sep=0pt
    },
    column 1/.style={nodes={fill=green!10}},
    row 1/.style={nodes={fill=green!10}},
  },
  square matrix/.default=0.8cm
}
\tikzset{slice/.style={path picture={
\stepcounter{slicenum}

\coordinate (O-\theslicenum) at ($(path picture bounding box.south
west)!0.5!(path picture bounding box.north east)$);
\pgfmathsetmacro{\myshift}{2*(1-sqrt(3)/2)*\side}
\coordinate (A-\theslicenum) at
($(O-\theslicenum)+(150:\side)+(0,{\myshift})$);
\coordinate (B-\theslicenum) at ($(O-\theslicenum)+(30:\side)+(0,{\myshift})$);
\coordinate (C-\theslicenum) at ($(O-\theslicenum)+(-90:\side)+(0,{\myshift})$);
\filldraw (A-\theslicenum) circle (1pt) (B-\theslicenum) circle (1pt) (C-\theslicenum) circle (1pt);
\foreach \x/\y/\z in {#1} {\draw[\z] (\x-\theslicenum) -- (\y-\theslicenum);}
}}}
\matrix[square matrix, box={3,5}] (A)
{
  & $a_0$ & $a_1$ & $a_2$ & $a_3$ & $a_4$ & $a_5$\\ 
$b_0$ & 0&0 &0 & $b_0$ &0 &0 \\ 
$b_1$ & $a_0$ & $a_1$ & $a_2$ & $m$ & $a_4$ & $a_5$ \\ 
$b_2$ &  0&  0& 0& $b_2$&0 &0   \\ 
$b_3$ &  0& 0 & 0 & $b_3$& 0& 0\\ 
$b_4$ &  0&  0& 0 &  $b_4$&0 &0 \\ 
$b_5$ &  0&  0& 0 &  $b_5$&0  &0 \\ 
};

\draw (A-1-1.north west)--(A-1-1.south east);
\node[below left=0mm and 0mm of A-1-1.north east, scale=0.85] {$Y_0$};
\node[above right=0mm and 0mm of A-1-1.south west, scale=0.85] {$Y_1$};
\node [draw,below=10pt] at (A.south) 
    { Assume $\delta=-2$, $a_3+b_1-1=m>0$ };
\end{tikzpicture}
\caption{Example construction of the matrix in  \textcolor{red}{Case 2} }
\label{fig:construction_matrix_case_2}
\end{figure}

Next, we show that there exists a joint distribution of $Y_1,Y_0$ such that $P(Y_1=i, Y_0=i-\delta)=\min\{P(Y_1=i), P( Y_0=i-\delta)\}=U_i$ for all $i$. We will construct a joint distribution of $Y_1, Y_0$ using the joint probability table explicitly. We first permute the joint probability table based on whether the upper bound is achieved at $P(Y_1=i)$ or $P(Y_0=i-\delta)$. Let $J_1$ be the set of $i$ such that $P(Y_1=i, Y_0=i-\delta)=P(Y_1=i)$ for $i\in J_1$ and $J_2$ be the set of $j$ such that $P(Y_1=j+\delta, Y_0=j)=P(Y_0=j)$ for $j\in J_2$. We can rearrange the joint probability table based on the partition of $J_1, J_2$. Let $|J_1|=n_1, |J_2|=n_2$. Let $N_1$ be the number of elements in the support of $Y_1$. Let $N_2$ be the number of elements in the support of $Y_0$. We permute the joint probability matrix of $Y_1, Y_0$ such that the first $n_1$ rows denotes the probability of elements in $J_1$ and the last $n_2$ rows denotes the probability of elements in $J_2$. Figure $\ref{Probability Table Permutation}$ shows an example of such permutation.
\begin{figure}[h!]
\centering
\begin{tikzpicture}
\tikzset{
  rows/.style 2 args={
    sub@rows/.style={row ##1 column #2/.style={nodes={very thick, rectangle,draw=red}}},
    sub@rows/.list={#1}
  },
  box/.style 2 args={
    sub@box/.style={rows={#1}{##1}},
    sub@box/.list={#2}
  }
}
\tikzset{%
square matrix/.style={
    matrix of nodes,
    column sep=-\pgflinewidth, 
    row sep=-\pgflinewidth,
    nodes in empty cells,
    nodes={draw,
      minimum size=#1,
      anchor=center,
      align=center,
      inner sep=0pt
    },
    column 1/.style={nodes={fill=green!10}},
    row 1/.style={nodes={fill=green!10}},
  },
  square matrix/.default=0.8cm
}
\tikzset{slice/.style={path picture={
\stepcounter{slicenum}
\coordinate (O-\theslicenum) at ($(path picture bounding box.south
west)!0.5!(path picture bounding box.north east)$);
\pgfmathsetmacro{\myshift}{2*(1-sqrt(3)/2)*\side}
\coordinate (A-\theslicenum) at
($(O-\theslicenum)+(150:\side)+(0,{\myshift})$);
\coordinate (B-\theslicenum) at ($(O-\theslicenum)+(30:\side)+(0,{\myshift})$);
\coordinate (C-\theslicenum) at ($(O-\theslicenum)+(-90:\side)+(0,{\myshift})$);
\filldraw (A-\theslicenum) circle (1pt) (B-\theslicenum) circle (1pt) (C-\theslicenum) circle (1pt);
\foreach \x/\y/\z in {#1} {\draw[\z] (\x-\theslicenum) -- (\y-\theslicenum);}
}}}
\matrix[square matrix, box={2,4},box={3,5},box={4,6},box={5,7}] (A)
{
  & $a_0$ & $a_1$ & $a_2$ & $a_3$ & $a_4$ & $a_5$\\ 
$b_0$ & & &$a_2$ & &0 & \\ 
$b_1$ & 0 & 0& 0& $b_1$ &0 & 0\\ 
$b_2$ &  &  & 0& &$a_4$ &   \\ 
$b_3$ &  0& 0 & 0 &0 & 0& $b_3$\\ 
$b_4$ &  &  & 0 &  &0 & \\ 
$b_5$ &  &  & 0 &  &0  & \\ 
};

\draw (A-1-1.north west)--(A-1-1.south east);
\node[below left=0mm and 0mm of A-1-1.north east, scale=0.85] {$Y_0$};
\node[above right=0mm and 0mm of A-1-1.south west, scale=0.85] {$Y_1$};
\node [draw,below=10pt] at (A.south) 
    { $A$ : Matrix \textcolor{red}{before} permutation};
\matrix[square matrix, box={2,4},box={3,5},box={4,6},box={5,7}] (C) at (6.5cm,0cm)
{
  & $a_0$ & $a_1$ & $a_3$ & $a_5$ & $a_2$ & $a_4$\\ 
$b_1$ & 0& 0&$b_1$ & 0&0 &0 \\ 
$b_3$ & 0 & 0& 0& $b_3$ &0 & 0\\ 
$b_0$ &  &  & & &$a_2$ &  0 \\ 
$b_2$ &  &  &  & & 0& $a_4$\\ 
$b_4$ &  &  &  &  &0 & 0\\ 
$b_5$ &  &  &  &  &0  & 0\\ 
};

\draw (C-1-1.north west)--(C-1-1.south east);
\node[below left=0mm and 0mm of C-1-1.north east, scale=0.85] {$Y_0$};
\node[above right=0mm and 0mm of C-1-1.south west, scale=0.85] {$Y_1$};
\node [draw,below=10pt] at (C.south) 
    { $A'$ : Matrix \textcolor{red}{after} permutation};
\end{tikzpicture}
\caption{Probability Table Permutation}
\label{Probability Table Permutation}
\end{figure}

By construction, the first $n_1$ rows and the last $n_2$ columns are filled with $0$'s and $U_i$'s. We have an empty $(N_1-n_1)\times(N_2\times n_2)$ sub-matrix to fill. We also obtain a new set of margin constraints on the sub-matrix by subtracting the marginals with what we have already filled. Notice that the row/column sum of these new constraints for the $(N_1-n_1)\times(N_2\times n_2)$ matrix is given by
\begin{align*}
s=1-\sum_{i\in J_1}P(Y_1=i)-\sum_{j\in J_2}P(Y_0=j).
\end{align*}
We can fill each entry of the sub-matrix by 
\begin{align*}
P(Y_1=i, Y_0=j)&= (P(Y_1=i)-I((i+\delta)\in J_2)P(Y_0=i+\delta))\\ &\times (P(Y_0=j)-I(j-\delta\in J_1)P(Y_1=j-\delta))/s.
\end{align*}
Thus, we construct a valid joint distribution of $Y_1, Y_0$ with $P(Y_1=i, Y_0=i-\delta)=\min\{P(Y_1=i), P( Y_0=i-\delta)\}=U_i$ for all $i$ and satisfies the given marginals. This concludes our proof of theorem \ref{pmf}. We note that the same proof technique can be used to obtain the pmf bounds on the sum of two discrete random variables given the marginals.

\begin{corollary}
In the case where the distribution of one potential outcome is degenerate, the individual treatment effect is fully identified. For example, if we assume $P(Y_0=0)=1$, then the upper and lower bounds in Theorem \ref{pmf} coincide and we can identify $P(Y_1-Y_0=\delta)$.
\end{corollary}
\begin{proof}
This follows from the fact that $[L_i, U_i]=[\max\{P(Y_1=i)+P( Y_0=i-\delta)-1,0\},\min\{P(Y_1=i), P( Y_0=i-\delta)\}]=[0,0]$ when $i-\delta\neq 0$ and $[L_i, U_i]=[\max\{P(Y_1=i)+P( Y_0=i-\delta)-1,0\},\min\{P(Y_1=i), P( Y_0=i-\delta)\}]=[P(Y_1=i),P(Y_1=i)]$ when $i-\delta =0$. Therefore, $\sum_iL_i=\sum_iU_i=P(Y_1=\delta)$. This means that in the case where $P(Y_0=0)=1$, the bounds in Theorem \ref{pmf} tell us that $P(Y_1-Y_0=\delta)=P(Y_1=\delta)$.
\end{proof}

\section{Discussion of ATE versus ITE}\label{sec:discuss_ate_ite}

In this section, we discuss the relationship between the average treatment effect (ATE) and the individual treatment effect (ITE). We focus on the implications for prediction intervals and hypothesis testing. The ATE is a well-defined parameter and is identifiable under standard assumptions. Consequently, as the sample size increases, confidence intervals for the ATE will eventually shrink to a point. In contrast, the individual treatment effect (ITE) is not a parameter in the classical sense: it is a random quantity defined at the unit level. Although prediction intervals for the ITE may become narrower with larger sample sizes, they do not, in general, converge to zero width. This reflects the inherent uncertainty in predicting individual level responses, even when the population-level effect is precisely estimated.

The Neyman null hypothesis posits a zero effect on average, while the Fisher null hypothesis posits no effect for \emph{every} individual. By construction, a true Fisher null implies the Neyman null; if all individual effects are zero, their average must be zero.  A true Fisher null also implies that $\{0\}$ will be a valid prediction interval for any level $\alpha$. A true Neyman null implies that confidence interval for ATE will contain zero with probability $(1-\alpha)\%$.

When performing hypothesis testing, $\{0\}$ being a valid prediction interval can be considered as evidence in support of the Fisher null while confidence interval does not contain $0$ can be considered as evidence against Neyman null. However, in finite samples, a failure to reject the Fisher null does not imply the ATE is in fact zero, and rejecting the Neyman null (i.e. finding a nonzero estimated ATE) does not automatically imply that $\{0\}$ cannot be a valid prediction interval.

 It is possible to construct a valid prediction interval for ITE as $\{0\}$, even when the data yield a confidence interval for the ATE that excludes zero. Formally, if $\alpha = 0.05$, it is possible that $95\%$ of the individuals have a zero treatment effect, while the ATE analysis, being an aggregate statement, detects a statistically meaningful difference overall from the rest of the $5\%$ population. This implies, perhaps paradoxically, we can have evidence against Neyman null but also evidence in support of Fisher null.

Non-individualized decision policies can consequently outperform individualized decision policies in some partial identification settings; see \cite{cui2021individualized}. Even with an ATE significantly different from zero, substantial uncertainty at the individual level (reflected in wide ITE intervals) can yield a scenario where assigning the same treatment to everyone is more reliable than any personalized rule that attempts to exploit covariate information. This outcome highlights how partial identification and weakly informative data can obscure the individual-level signal, despite showing clear evidence of an overall effect.

\subsection{Synthetic data example}

Consider a randomized trial for $n=50000$ participants with binary outcomes, divided into treatment and control groups using a coin flip. Suppose that treatment cures $3\%$ of the population, hurts $1\%$ of the population, while having no effect on $96\%$ of the populations. We observe the following marginals in Table \ref{tab:example_ate_ite}.

\begin{table}[h]
    \centering
    \begin{tabular}{r|cc}
     & $\hat{P}(Y\!=\!0\mid D\!=\!0) = 0.9896$ & $\hat{P}(Y\!=\!1\mid D\!=\!0)= 0.0104$\\[4pt]
     \hline
     &\\[-8pt]
     $\hat{P}(Y\!=\!0\mid D\!=\!1)=0.9697$ & $P($NR$)$ & $P($HU$)$ \\[4pt]
      $\hat{P}(Y\!=\!1\mid D\!=\!1)=0.0303$ & $P($HE$)$ &$P($AR$)$ \\[4pt]
    \end{tabular}
    \caption{Marginal distributions estimated from the numerical simulation.}
    \label{tab:example_ate_ite}
\end{table}

We estimate the average treatment effect to be $0.0198$ with $95\%$ confidence interval $[0.0174 , 0.0223]$. This yields a confidence interval for the average treatment effect excluding zero. However, based on the result in Section \ref{singleton}, we will conclude a $95\%$ prediction interval for ITE as a singleton $\{0\}$.

Although this situation might seem paradoxical, it simply reflects the fact that population-level inferences about the mean effect can diverge from inferences about individual-level effects under uncertainty. On the one hand, one can accumulate enough information to claim that the \emph{average} effect is nonzero, while on the other hand, one lacks the precision required to identify which individuals truly benefit.

\subsection{Further discussions}
Even with unconfoundedness (or randomization), the difficulty of estimating ITE lies in the unknown structure of potential outcomes $Y_1, Y_0$. If $Y_1$ and $Y_0$ are independent, then the fundamental problem of causal inference does not exist. For example, the introduction section of \cite{yin2022conformal} can be confusing. Similarly, the section on numerical experiments in \cite{lei2021conformal} also assumes the independence of $Y_1$ and $Y_0$.

A complementary strategy under partial identification is to condition on covariates $X$, seeking subpopulations where treatment effects are more homogeneous. Our results can be extended to this scenario when the treatment is conditionally independent of the potential outcomes. For example, if there exist covariates for which $P(Y=1 \mid D=1,X)$ or $P(Y=1 \mid D=0,X)$ are close to $0$ or $1$, the uncertainty about individual-level effects in that subgroup may be greatly reduced, enabling tighter bounds on $P\bigl(Y_{1}-Y_{0} \in [L,R] \mid X\bigr)$. In Appendix \ref{sec:pi_ite_covariates}, we consider the ITE prediction intervals in the setting of randomized experiments where additional covariates are observed for each individual and provide an example of how the conditional ITE prediction interval can be even wider compared to the normal ITE prediction interval. We also need to be cautious with such conditional individual treatment effect prediction intervals as they could disagree with the conditional average treatment effect (CATE) estimation results. 

\section{Acknowledgments}
We thank James M. Robins, Carlos Cinelli, Yanqin Fan, Ting Ye and Gary Chan for valuable input and discussion.

\newpage
\singlespacing
\bibliography{ITE_prediction_intervals_and_sharp_bounds_arxiv}

\newpage
\doublespacing
\begin{appendices}
\section{When is a given prediction interval valid?}\label{app:when_valid}
\subsection{When is $\{0\}$ a valid prediction interval for the ITE?}

A valid interval consists of $\{0\}$ when the proportion of people of type Always Recover plus the proportion of type
Never Recover is known to be greater than or equal to $(1-\alpha)$.

This implies that:
\begin{align*}
&1-P(Y\!=\!1 \mid D\!=\!0)-P(Y\!=\!1 \mid D\!=\!1)\\
& + 2 \max\left\{0,(P(Y\!=\!1 \mid D\!=\!0)+P(Y\!=\!1 \mid D\!=\!1)) - 1\right\}
\geq (1-\alpha),
\end{align*}
equivalently,
\begin{align*}
&\max\left\{ 1-P(Y\!=\!1 \mid D\!=\!0)-P(Y\!=\!1 \mid D\!=\!1),\right.\\  
&\left. P(Y\!=\!1 \mid D\!=\!0)+P(Y\!=\!1 \mid D\!=\!1) - 1\right\} \geq (1-\alpha),
\end{align*}
or in other words, either
\begin{align*}
P(Y\!=\!1 \mid D\!=\!0)+ P(Y\!=\!1 \mid D\!=\!1)\leq \alpha,
\end{align*}
or 
\begin{align*}
P(Y\!=\!0 \mid D\!=\!0)+ P(Y\!=\!0 \mid D\!=\!1)\leq \alpha.
\end{align*}

\subsubsection{When is $\{1\}$ a valid prediction interval for the ITE?}

A valid interval consists of $\{1\}$ when the proportion of people of type Helped is always greater than or equal to $(1-\alpha)$.

This will occur when:
\begin{align*}
P(Y=1\mid D=1) - P(Y=1\mid D=0) \geq (1-\alpha).
\end{align*}
In other words, the Average Treatment Effect is required to be greater than $(1-\alpha)$.

\subsubsection{When is $\{-1\}$ a valid prediction interval for the ITE?}

A valid interval consists of $\{-1\}$ when the proportion of people of type Hurt is always greater than or equal to $(1-\alpha)$.

This will occur when:
\begin{align*}
P(Y=1\mid D=0) - P(Y=1\mid D=1) \geq (1-\alpha).
\end{align*}
In other words, the Average Treatment Effect is required to be less than $-(1-\alpha)$.

\subsection{When is the valid prediction interval for the individual treatment effect non-negative/non-positive?}

We now consider when a valid $(1-\alpha)\%$ prediction interval rules out a negative/positive result. There are two cases to consider:

\subsubsection{When is $[0,1]$ a valid prediction interval for the ITE?}
A valid interval consists of $[0,1]$ when the proportion of people of type Helped plus proportion of people of type Always Recover plus proportion of people of type Never Recover is always greater than or equal to $(1-\alpha)$. In other words, the proportion of people of type Hurt should always be less than $\alpha$. This implies that
\begin{align*}
P(Y=1\mid D=0)-t< \alpha
\end{align*}
for all $t$. That is, 
\begin{align*}
P(Y=1\mid D=0)-\max\left\{0,(P(Y\!=\!1 \mid D\!=\!0)+P(Y\!=\!1 \mid D\!=\!1)) - 1\right\}< \alpha,
\end{align*}
equivalently, either
\begin{align*}
P(Y=1\mid D=0)< \alpha,
\end{align*}
or 
\begin{align*}
1-P(Y=1\mid D=1)< \alpha.
\end{align*}
\subsubsection{When is $[-1,0]$ a valid prediction interval for the ITE?}
A valid interval consists of $[-1,0]$ when the proportion of people of type Hurt plus proportion of people of type Always Recover plus proportion of people of type Never Recover is always greater than or equal to $(1-\alpha)$. In other words, the proportion of people of type Helped should always be less than $\alpha$. This implies that
\begin{align*}
P(Y=1\mid D=1)-t< \alpha
\end{align*}
for all $t$. That is, 
\begin{align*}
P(Y=1\mid D=1)-\max\left\{0,(P(Y\!=\!1 \mid D\!=\!0)+P(Y\!=\!1 \mid D\!=\!1)) - 1\right\}< \alpha,
\end{align*}
equivalently, either
\begin{align*}
P(Y=1\mid D=1)< \alpha,
\end{align*}
or 
\begin{align*}
1-P(Y=1\mid D=0)< \alpha.
\end{align*}

\section{Necessary conditions for a given prediction interval to be valid and of minimal length}\label{app:necessary_conditions}

\subsection{Necessary conditions for $\{1\}$ being a valid prediction interval}
Suppose we know $\{1\}$ is a valid prediction interval, what are the necessary conditions on the marginal distributions $P(Y=1\mid D=0)$ and $P(Y=1\mid D=1)$? This means that there exists $t$ such that the proportion of Helped is greater than $1-\alpha$. That is:
\begin{align}
    P(Y=1\mid D=1)-t\geq 1-\alpha\label{eq:1_mix}
\end{align}
for some 
\begin{align*}
    \max\{0,P(Y=1\mid D=1)+P(Y=1\mid D=0)-1\}\leq t\leq\min\{P(Y=1\mid D=1),P(Y=1\mid D=0)\}.
\end{align*}
We obtain the constraints
\begin{align*}
    P(Y=1\mid D=1)&\geq 1-\alpha;\\
    P(Y=1\mid D=0)&\leq \alpha.
\end{align*}
For any marginal distribution that satisfies these constraints, there exists a $t$ that satisfies (\ref{eq:1_mix}).
\subsection{Necessary conditions for $\{-1\}$ being a valid prediction interval}
Similarly, for
\begin{align*}
    P(Y=1\mid D=0)&\geq 1-\alpha,\\
    P(Y=1\mid D=1)&\leq \alpha,
\end{align*}
there exists $t$ such that 
\begin{align*}
    P(Y=1\mid D=0)-t\geq 1-\alpha.
\end{align*}
\subsection{Necessary conditions for $\{0\}$ being a valid prediction interval}
Suppose we know $\{0\}$ is a valid prediction interval,  what are the necessary conditions on the marginal distributions $P(Y=1\mid D=0)$ and $P(Y=1\mid D=1)$? This means the proportion of Never Recover plus Always Recover is greater than or equal to $1-\alpha$. That is:
\begin{align}
1 - P(Y=1 \mid D=1) - P(Y=1 \mid D=0) + 2t \geq 1 - \alpha. \label{eq:0PI}
\end{align}
For any
\begin{align*}
    |P(Y=1\mid D=1)-P(Y=1\mid D=0)|\leq \alpha,
\end{align*}
there exists a $t$ that satisfies $(\ref{eq:0PI})$.
\subsection{Necessary conditions for $[-1,1]$ being the best prediction interval}
Suppose we know $[-1,1]$ is the best prediction interval we can get. This means there exists $t$ such that the proportion of people of type Hurt and the proportion of people of type Helped are both greater than or equal to $\alpha$. That is:
\begin{align*}
    P(Y=1\mid D=1)-t\geq \alpha;\\
    P(Y=1\mid D=0)-t\geq \alpha.
\end{align*}
We obtain the constraints
\begin{align}
    \alpha\leq P(Y=1\mid D=1)&\leq 1-\alpha; \label{eq:-1to11}\\
    \alpha \leq P(Y=1\mid D=0)&\leq 1-\alpha.\label{eq:-1to12}
\end{align}
For any marginal distribution that satisfies these constraints, there exists a $t$ that satisfies both $(\ref{eq:-1to11}), (\ref{eq:-1to12})$.
\subsection{Necessary conditions for $[-1,0]$ being the best prediction interval}
Suppose we know $[-1,0]$ is the best prediction interval we can get. First, it needs to be valid. This means there exists $t$ such that the proportion of people of type Hurt, Always Recovered and Never Recovered needs to be greater than or equal to $1-\alpha$. That is, the proportion of people of type Helped is less than or equal to $\alpha$. Then it needs to be best (we cannot conclude $\{-1\}, \{0\}$ as prediction intervals and $[-1,0]$ has better coverage than $[0,1]$), that means the proportion of people of type Hurt is greater than or equal to the proportion of people of type Helped, the proportion of people of type Hurt is less than $1-\alpha$, the proportion of people of type Always Recover plus the proportion of people of type Never Recover is less than $1-\alpha$. That is:
\begin{align}
    P(Y=1\mid D=1)-t\leq \alpha; \label{eq:-1to01}\\
    P(Y=1\mid D=0)-t\geq P(Y=1\mid D=1)-t;\\
    P(Y=1\mid D=0)-t<1-\alpha;\\
    1-P(Y=1\mid D=1)-P(Y=1\mid D=0)+2t< 1-\alpha. \label{eq:-1to03}
\end{align}
We obtain the constraints
\begin{align*}
P(Y=1\mid D=0)&\geq P(Y=1\mid D=1);\\
    1-P(Y=1\mid D=0)+P(Y=1\mid D=1)&> \alpha;\\
    P(Y=1\mid D=0)+P(Y=1\mid D=1)&> \alpha;\\
    2-P(Y=1\mid D=0)-P(Y=1\mid D=1)&> \alpha.
\end{align*}
For any marginal distribution that satisfies these constraints, there exists a $t$ that satisfies both $(\ref{eq:-1to01})- (\ref{eq:-1to03})$.

\subsection{Necessary conditions for $[0,1]$ being a best prediction interval}
Similarly, suppose we know $[0,1]$ is the best prediction interval we can get. This means there exists $t$ such that the proportion of people of type Hurt is less than or equal to $\alpha$, the proportion of people of type Helped is greater than or equal to the proportion of people of type Hurt, the proportion of people of type Helped is less than $1-\alpha$, the proportion of people of type Always Recover plus the proportion of people of type Never Recover is less than $1-\alpha$. That is:
\begin{align}
    P(Y=1\mid D=0)-t\leq \alpha; \label{eq:0to11}\\
    P(Y=1\mid D=1)-t\geq P(Y=1\mid D=0)-t;\\
    P(Y=1\mid D=1)-t<1-\alpha;\\
    1-P(Y=1\mid D=1)-P(Y=1\mid D=0)+2t< 1-\alpha. \label{eq:0to13}
\end{align}
We obtain the constraints
\begin{align*}
P(Y=1\mid D=1)&\geq P(Y=1\mid D=0);\\
    1-P(Y=1\mid D=1)+P(Y=1\mid D=0)&> \alpha;\\
    P(Y=1\mid D=0)+P(Y=1\mid D=1)&> \alpha;\\
    2-P(Y=1\mid D=0)-P(Y=1\mid D=1)&> \alpha.
\end{align*}
For any marginal distribution that satisfies these constraints, there exists a $t$ that satisfies both $(\ref{eq:0to11})- (\ref{eq:0to13})$.

\section{Sharp Bounds on the Cumulative Distribution of Individual Treatment Effect}\label{sec:sharp_CDF_bounds}
 \cite{fan2010sharp} state bounds on the distribution of treatment effect that are sharp when the distribution of the outcome is absolutely continuous with respect to Lebesgue measure. \cite{zhang2024bounds} modify the bounds in \cite{fan2010sharp} so that they are sharp when the distribution function of the outcome is not continuous. Let $Y_1$ and $Y_0$ be potential outcomes of each individual in the population. To simplify notation, define the individual treatment effect $\Delta=Y_1-Y_0$. Let $F_1, F_0$ be the cumulative distribution function (cdf) on $Y_1, Y_0$ respectively. Let $F_\Delta(\cdot)$ be the cdf for $\Delta$.
\begin{theorem}[\citeauthor{zhang2024bounds}  ]\label{bounds}
For any given value $\delta$, sharp bounds on $F_\Delta(\delta)$ are given by
 \begin{align}
F^L(\delta)&=\sup_y\max\{F_1(y)-P(Y_0<y-\delta),0\}\nonumber \\
&= \sup_y\max\{F_1(y)-F_0(y-\delta)+P(Y_0=y-\delta),0\};\label{eq:correct-upper}\\[16pt]
F^U(\delta)&=1+\inf_y \min\{F_1(y)-P(Y_0<y-\delta),0\}\nonumber \\
&=  1+\inf_y \min\{F_1(y)-F_0(y-\delta)+P(Y_0=y-\delta),0\}.\label{eq:correct-lower}
\end{align}
\end{theorem}
Further details on these bounds for individual treatment effects, along with their implications, are provided in \cite{zhang2024bounds}. The proofs therein leverage key results from \cite{makarov1982estimates}, \cite{ruschendorf1982random}, and \cite{frank1987best}. Notably, \cite{zhang2024bounds} correct aspects of earlier findings reported in \cite{williamson1990probabilistic} and \cite{fan2010sharp}.

\begin{remark}[]
As discussed in \cite{zhang2024bounds}, the lower bound proposed in Lemma $2.1$ of \cite{fan2010sharp} is valid but not necessarily sharp and the upper bound might not be valid. The bounds in Theorem \ref{bounds} can be applied to more general settings when the outcome is discrete. For example, the ordinal treatment effects discussed in \cite{lu2018treatment} and \cite{huang2017inequality}. In general, the pmf bounds calculated using the cdf bounds are not sharp. We will discuss sharp pmf bounds in Section \ref{sec:sharp_PMF_bounds}.
\end{remark}

\section{Understanding prediction intervals for individual treatment effect when the outcome is ordinal} \label{section:high_dimension}
We now return to the question originally considered in section \ref{sec:ITE_prediction_binary} on prediction intervals with ordinal outcome.
\subsection{When will the prediction interval be trivial?}
What is the condition on the observed distribution under which without additional knowledge the only interval that we can be sure is valid is a trivial prediction interval? Let's assume $Y_1,Y_0$ takes value $\{0,1,2,...,n-1\}$. Then a trivial prediction interval would be $[-n+1, n-1]$. If it is possible that $P(\mathrm{ITE}=n-1)$ and $P(\mathrm{ITE}=-n+1)$ can both be greater than $\alpha$, then we can only provide a trivial prediction interval.  A necessary condition for $P(\mathrm{ITE}=n-1)$ and $P(\mathrm{ITE}=-n+1)$ to be greater than $\alpha$ is that the Fr\'{e}chet inequality upper bounds for both $P(Y_0=n-1,Y_1=0)$ and $P(Y_0=0,Y_1=n-1)$ are greater than $\alpha$. That is,
\begin{align*}
\min\{P(Y_1=0),P(Y_0=n-1)\}&> \alpha \\
\text{and}\quad \min\{P(Y_1=n-1),P(Y_0=0)\}&> \alpha. 
\end{align*}

Here we only consider prediction intervals. We could also consider non-overlapping prediction intervals or prediction sets.

\subsection{When can we tell that the prediction interval need not include $0$?}
We want to answer the question: what is the marginal condition that allow us to conclude the prediction interval not include $0$. A necessary condition for this is $P(\mathrm{ITE}=0)<\alpha$. Based on Theorem \ref{pmf}, the necessary condition requires that
\begin{align*}
\sum_i\min\{P(Y_1=i), P( Y_0=i)\}<\alpha.
\end{align*}
Alternatively, if we interpret the question as whether the left endpoint of the prediction interval is less than or equal to $0$ and the right endpoint is greater than or equal to $0$, then the same result from Section~\ref{subsec:away_0_continuous} applies. In particular, the conditions derived using cdf bounds for the continuous outcome setting can also be applied to the ordinal outcome setting, allowing us to assess whether the prediction interval needs to contain zero.

\section{Prediction Intervals for ITE with Covariates}\label{sec:pi_ite_covariates}
In this section, we illustrate the construction of ITE prediction intervals in the context of randomized experiments where additional covariates are observed for each individual. As we will demonstrate through an example, the prediction interval for the ITE conditional on covariates is not necessarily shorter than its marginal counterpart. In general, ITE prediction intervals and conditional prediction intervals are not directly comparable, as they address different inferential targets and rely on distinct sources of variability.
\subsection{Problem Setup}
Consider i.i.d. random samples $\left\{\left(D_i, X_i, Y_i\right)\right\}_{i=1}^n$ of $n$ individuals, where $D_i \in\{0,1\}$ is a binary treatment indicator, $X_i=\left(X_{i 1}, X_{i 2}, \ldots, X_{i p}\right)^T \in \mathbb{R}^p$ is a vector of observed covariates for each individual $i$, and $Y_i \in \mathbb{R}$ is the observed outcome for individual $i$ under the potential outcome framework.

We assume that each individual receives the treatment independently with equal probability $P(D_i=1 \mid X_i)=\pi$ for all $i$, where $0<\pi<1$ is a known constant. To simplify notation, we suppress the subscript $i$ in what follows. Note that this is equivalent to saying the treatment is randomized regardless of the observed covariates. Each individual has two potential outcomes $Y_1$ and $Y_0$ and we observe $Y = DY_1+(1-D)Y_0$.  We assume the stable unit treatment value assumption (SUTVA) that there is a single version of each treatment/control and no interference among the subjects. The individual treatment effect is defined as:
$$
\mathrm{ITE}=Y_1-Y_0
$$
\subsection{Example where a subset is more homogeneous}
\label{sec:example-subset-homogeneous}

Here we present a simple numeric illustration of a randomized experiment with two binary covariates, 
$X_1, X_2 \in \{0,1\}$, and a binary outcome \(Y\). 
Suppose there are $10000$ subjects in the treatment arm and $10000$ subjects in the control arm, allocated 1:1 at random. 
Table~\ref{tab:example-subset} displays how many subjects fall into each \((X_1,X_2)\) cell, along with the number of observed outcome \(Y=1\) and the corresponding probability that $Y=1$.

\begin{table}[htbp]
\centering
\begin{tabular}{ccrcrcrr}
\hline
\multicolumn{2}{c}{Covariates} & \multicolumn{3}{c}{Treatment ($D=1$)} & \multicolumn{3}{c}{Control ($D=0$)} \\
\cline{1-2}\cline{3-5}\cline{6-8}
\(X_1\) & \(X_2\) & \(n_{T}\) & \(Y=1\) & \(\hat{p}_{T}\) & \(n_{C}\) & \(Y=1\) & \(\hat{p}_{C}\) \\
\hline
0 & 0 & 3000 & 1200 & 0.40 & 3000 & 1000 & 0.33 \\
0 & 1 & 3000 & 1000 & 0.33 & 3000 &  800 & 0.27 \\
1 & 0 & 3000 & 3000 & 1.00 & 3000 &  0 & 0.00 \\
1 & 1 & 1000 &  800 & 0.80 & 1000 &  200 & 0.2 \\
\hline
\textbf{Total} & -- 
 & 10000 & 6000 & 0.60 
 & 10000 & 2000 & 0.20 \\
\hline
\end{tabular}
\caption{Synthetic data in which the treatment arm has overall $60\%$ of $Y=1$ and the control arm has $20\%$ of $Y=1$. 
Within certain subgroups (e.g.\ conditioning on \(X_1=1\)), the treatment effect is more homogeneous, but conditioning further (e.g.\ on \((X_1=1, X_2=1)\)) breaks the homogeneity again.}\label{tab:example-subset}
\end{table}

\subsection{Key observations}
\begin{itemize}
\item Without covariate adjustment, the treatment arm has probability $Y=1$ equal to $0.6$ and the control arm has probability $Y=1$ equal to $0.2$. 

\item Condition on \(X_1=1\), in treatment, among those with \(X_1=1\) (rows with \((X_1=1, X_2=0)\) and \((X_1=1, X_2=1)\)), we have 
  \(n_T = 3000 + 1000 = 4000\) individuals, with \(3000 + 800 = 3800\) individuals with $Y=1$, i.e.\ \(95\%\) of individuals with $X_1 =1$ in the treatment arm have positive outcomes.  
  In control arm, among \(X_1=1\), \(n_C = 3000 + 1000 = 4000\) with \(0 + 200 = 200\) individuals with $Y=1$, i.e.\ \(5\%\) of individuals with $X_1 =1$ in the control arm have positive outcomes.
\item Condition on \((X_1=1, X_2=1)\), we are restricting further to just the last row. In this case,  
\(80\%\) of individuals in the treatment arm have positive outcomes while \(20\%\) of individuals in the control arm have positive outcomes.

\end{itemize}

\subsection{Constructing ITE prediction intervals}
Now given this experiment data, suppose we want to construct a $90\%$ prediction interval for this people with $X_1=1, X_2=1$. First, under randomization we have:
$$P(Y\!=\!i\mid D\!=\!j,  X_1\!=\!1) = P( Y_j\!=\!i\mid X_1 = 1)$$
We can write the following two-way table:
\begin{center}
\begin{table}[h!]
\scalebox{0.85}{
\begin{tabular}{r|ccc}
 & $\hat{P}(Y\!=\!0\mid D\!=\!0, X_1 \!=\!1)=0.95$ & $\hat{P}(Y\!=\!1\mid D\!=\!0, X_1 \!=\! 1)=0.05$ \\[4pt]
 \hline
 &\\[-5pt]
 $\hat{P}(Y\!=\!0\mid D\!=\!1, X_1 \!=\!1)=0.05$ & $P(Y_1 = 0, Y_0 = 0\mid X_1=1)$& $P(Y_1 = 0, Y_0 = 1\mid X_1=1)$\\[4pt]
  $\hat{P}(Y\!=\!1\mid D\!=\!1, X_1 \!=\!1) = 0.95$ & $P(Y_1 = 1, Y_0 = 0\mid X_1=1)$ &$P(Y_1 = 1, Y_0 = 1\mid X_1=1)$ \\[4pt]
    \end{tabular}
\bigskip
}
\caption{Binary Treatment and Outcome Model with condition on $X_1 =1$}
\end{table}
\end{center}
By Frechet inequality, we can conclude that
$$P\left(Y_1-Y_0\in \{1\}\mid X_1=1\right) = P(Y_1 = 1, Y_0 = 0\mid X_1=1) \geq 0.95+0.95-1= 0.9$$
Note that if we instead condition on both $X_1 =1$ and $X_2 =1$, we will get: 
\begin{center}
\begin{table}[h]
\scalebox{0.85}{
\begin{tabular}{r|ccc}
 & $\hat{P}(Y\!=\!0\mid D\!=\!0, X_1 \!=\!1,X_2 \!=\!1)=0.8$ & $\hat{P}(Y\!=\!1\mid D\!=\!0, X_1 \!=\! 1,X_2 \!=\!1)=0.2$ \\[4pt]
 \hline
 &\\[-5pt]
 $\hat{P}(Y\!=\!0\mid D\!=\!1, X_1 \!=\!1,X_2 \!=\!1)=0.2$ & $P(Y_1 = 0, Y_0 = 0\mid X_1=1,X_2 \!=\!1)$& $P(Y_1 = 0, Y_0 = 1\mid X_1=1,X_2 \!=\!1)$\\[4pt]
  $\hat{P}(Y\!=\!1\mid D\!=\!1, X_1 \!=\!1,X_2 \!=\!1) = 0.8$ & $P(Y_1 = 1, Y_0 = 0\mid X_1=1,X_2 \!=\!1)$ &$P(Y_1 = 1, Y_0 = 1\mid X_1=1,X_2 \!=\!1)$ \\[4pt]
    \end{tabular}
\bigskip
}
\caption{Binary Treatment and Outcome Model condition on $X_1 =1$ and $X_2 = 1$}
\label{table_binary_cond_X1X2}
\end{table}
\end{center}
Based on this table, we can only conclude that:
$$P\left(Y_1-Y_0\in \{-1,0,1\} \mid X_1=1,X_2 =1 \right) \geq 0.9$$
since the proportion of $\mathrm{ITE} =0$ and $\mathrm{ITE} =-1$ condition on $X_1=1,X_2=1$ can both be greater than $0.2$.
\subsection{Discussion of implications}
In this example, it maybe paradoxical that for an individual with $X_1 = 1 , X_2 =1$, we can conclude that this person's individual treatment effect is $1$ with probability $90\%$ condition only on $X_1=1$ but instead get a trivial prediction interval condition on $X_1 =1$ and $X_2=1$. What extra information does $X_2$ give? Additionally conditioning on $X_2=1$ picks out a subgroup in which the treatment effect is not as tightly concentrated as it is in the overall $X_1=1$ group. That would indicate that $X_2$ is related to heterogeneity in the effect. Conditioning on $X_2=1$ and $X_1=1$ defines a subgroup in which the treatment effect is more heterogeneous. In this case, do we really want to make decisions based on the narrower prediction interval? And as we can see, prediction intervals for ITE may not shrink if we condition on more covariates.
\subsection{Relationship between conditional ITE and ITE}
We will consider the simple case where $X_1\in \{0,1\}$. First, we have
\[
P\bigl(Y_1-Y_0 = 1\bigr)
=
P\bigl(Y_1-Y_0 = 1,\,X_1=1\bigr)
+
P\bigl(Y_1-Y_0 = 1,\,X_1=0\bigr).
\]
By the law of total probability, we can also write
\begin{align*}   
P\bigl(Y_1-Y_0 = 1\bigr)&
=
P\Bigl(Y_1-Y_0 = 1 \,\bigm\vert\, X_1=1\Bigr)\,P(X_1=1)
\\&+
P\Bigl(Y_1-Y_0 = 1 \,\bigm\vert\, X_1=0\Bigr)\,P(X_1=0).
\end{align*}
Thus, even if 
\[
P\Bigl(Y_1-Y_0 = 1 \,\bigm\vert\, X_1=1\Bigr) \geq 0.9,
\]
it does \textbf{not} necessarily imply 
\[
P\Bigl(Y_1-Y_0 = 1\Bigr) \geq 0.9,
\]
because the overall (unconditional) probability also depends on $P(X_1=1)$ and $P(Y_1-Y_0=1 \mid X_1=0)$. If $P(X_1=1)$ is small or if $P(Y_1-Y_0=1 \mid X_1=0)$ is small, the unconditional probability may be much less than $0.9$.

\end{appendices}

\end{document}